\documentclass[12pt,a4paper]{article}

\pdfoutput=1

\usepackage[margin=1in]{geometry}
\usepackage[
  bookmarks=true,
  bookmarksnumbered=true,
  bookmarksopen=true,
  pdfborder={0 0 0},
  breaklinks=true,
  colorlinks=true,
  linkcolor=black,
  citecolor=black,
  filecolor=black,
  urlcolor=black,
]{hyperref}
\usepackage{amssymb}
\usepackage{amsfonts}
\usepackage[round]{natbib}
\usepackage{amsmath}
\usepackage{amsthm}
\usepackage{mathtools}
\usepackage{nicefrac}
\usepackage{enumitem}
\usepackage[capitalize]{cleveref}
\usepackage{color}

\usepackage{fnbreak}
\allowdisplaybreaks

\newcommand{\crefpart}[2]{\cref{#1}(\labelcref{#1-#2})}
\newcommand{\crefsubpart}[3]{\cref{#1}(\labelcref{#1-#2})(\labelcref{#1-#2-#3})}
\newcommand{\crefsubpartonly}[3]{\cref{#1-#2}(\labelcref{#1-#2-#3})}
\newcommand{\refintitle}[1]{\texorpdfstring{\ref{#1}}{\ref*{#1}}}

\newlist{parts}{enumerate}{2}
\Crefname{partsi}{Part}{Parts}
\setlist[parts,1]{label=\alph*.,ref=\alph*}
\setlist[parts,2]{label=\roman*.,ref=\roman*}

\newlist{steps}{enumerate}{2}
\Crefname{stepsi}{Step}{Steps}
\Crefname{stepsii}{Step}{Steps}
\setlist[steps,1]{label=\arabic*.,ref=\arabic*}
\setlist[steps,2]{label=(\alph*),ref=\arabic{stepsi}(\alph*)}

\newtheorem{theorem}{Theorem}[section]
\newtheorem{lemma}[theorem]{Lemma}
\newtheorem{proposition}[theorem]{Proposition}
\newtheorem{corollary}[theorem]{Corollary}

\theoremstyle{definition}
\newtheorem{remark}[theorem]{Remark}
\newtheorem{definition}[theorem]{Definition}
\newtheorem{example}[theorem]{Example}

\DeclareMathOperator{\poly}{poly}
\DeclareMathOperator{\Rev}{Rev}
\DeclareMathOperator{\supp}{supp}
\DeclareMathOperator*{\EE}{\mathbb{E}}
\DeclareMathOperator*{\PP}{\mathbb{P}}
\DeclareMathOperator*{\Max}{Max}
\DeclareMathOperator{\MAX}{MAX}
\DeclareMathOperator{\APPROX}{APPROX}
\DeclareMathOperator{\freq}{freq}

\newcommand{\OPT}{\mathit{OPT}}
\newcommand{\RR}{\mathbb{R}_+}
\newcommand{\NN}{\mathbb{N}}
\newcommand{\eqdef}{\triangleq}
\newcommand{\epsfloor}[1]{\lfloor#1\rfloor_\varepsilon}
\newcommand{\epsceil}[1]{\lceil#1\rceil_\varepsilon}

\hypersetup{
  pdftitle       = {Efficient Empirical Revenue Maximization in Single-Parameter Auction Environments},
  pdfauthor      = {Yannai A. Gonczarowski <yannai@gonch.name> and Noam Nisan <noam.nisan@gmail.com>},
}

\title{Efficient Empirical Revenue Maximization in Single-Parameter Auction Environments}
\author{Yannai A. Gonczarowski\thanks{Einstein Institute of Mathematics, Rachel \& Selim Benin School of Computer Science \& Engineering and Federmann Center for the Study of Rationality, The Hebrew University of Jerusalem, Israel; and Microsoft Research, \emph{E-mail}: \href{mailto:yannai@gonch.name}{yannai@gonch.name}.}
\and
Noam Nisan\thanks{Rachel \& Selim Benin School of Computer Science \& Engineering and Federmann Center for the Study of Rationality, The Hebrew University of Jerusalem, Israel; and Microsoft Research, \emph{E-mail}: \href{mailto:noam.nisan@gmail.com}{noam.nisan@gmail.com}.}
}
\date{April 9, 2017}

\begin{document}

\maketitle

\begin{abstract}
We present a polynomial-time algorithm that, given samples from the unknown valuation distribution of each bidder, learns an auction that approximately maximizes the auctioneer's revenue in a variety of single-parameter auction environments including matroid environments, position environments, and the public project environment. The valuation distributions may be arbitrary bounded distributions (in particular, they may be irregular, and may differ for the various bidders), thus resolving a problem left open by previous papers. The analysis uses basic tools, is performed in its entirety in value-space, and simplifies the analysis of previously known results for special cases. Furthermore, the analysis extends to certain single-parameter auction environments where precise revenue maximization is known to be intractable, such as knapsack environments.
\end{abstract}

\section{Introduction}\label{intro}

We start by considering auctions of a single item to $n$ bidders (we later extend our results beyond the single-item auction environment).
The value (i.e., maximum willingness to pay) $v_i$ for the item of each bidder $i$ is distributed according to some {\em unknown} (perhaps irregular) distribution~$F_i$
that is supported on a known bounded interval $[0,H]$, and the values of the bidders are drawn independently of each other.
An auctioneer is given $t=\poly(H,n,\nicefrac{1}{\varepsilon})$ random samples
from each $F_i$, and aims to design an auction that, with high probability, approximately
maximizes her revenue from (future bidders with values drawn from) $F_1 \times \cdots \times F_n$. (I.e., the goal is to PAC learn a revenue-maximizing auction.)

The most natural approach would be to take, for each bidder $i$, the empirical distribution
$\hat{F}_i$ defined as the uniform distribution over the (multi)set of the $t$ samples from~$F_i$, and
design an auction that maximizes the empirical revenue, i.e., the revenue from $\hat{F}_1 \times \cdots \times \hat{F}_n$.
One would hope that the resulting auction would approximately maximize the revenue from 
$F_1 \times \cdots \times F_n$ as well.

Unfortunately, such an approach may be prone to overfitting the samples, in which case the empirical-revenue-maximizing auction may
perform poorly on $F_1 \times \cdots \times F_n$. Therefore, recent papers have focused
on designing auctions that approximately maximize the empirical revenue, in a way that is robust against overfitting.\footnote{It is nonetheless still not entirely clear whether overfitting can actually occur with nonnegligible probability for bounded valuation distributions.}
These papers combine the delicate understanding of the nature
of optimal auctions due to \cite{Myerson} with various techniques of 
converting an auction that is good for variants of the empirical
distribution to one that is good for the original distribution, in
several special cases: a single buyer with a regular value distribution \citep{Huang-Mansour-Roughgarden},
multiple buyers with regular distributions (\citealp{Cole-Roughgarden}; \citealp{Devanur-Huang-Psomas}\footnote{The authors of \cite{Devanur-Huang-Psomas} have recently notified us (personal communication, April~2017) that it turns out that their techniques can be generalized to also yield a result similar to our \cref{intro-empirical}, along with a few special cases of our \cref{intro-empirical-single-parameter}. The updated version of their paper is now available on \texttt{arXiv.org}.}), and
multiple buyers with i.i.d.\ irregular distributions \citep{ironing-in-the-dark}.  The general case (i.e., that of non-i.i.d., possibly irregular, distributions) remained an open problem \citep{jamie-tutorial}, even though it was shown that, information-theoretically, it should be possible \citep{Morgenstern-Roughgarden}.

\medskip

In this paper, we resolve the above-described open problem by showing how to take the auction that is optimal for the empirical sampled distribution $\hat{F}_1 \times \cdots \times \hat{F}_n$, and efficiently \!\emph{``round''} it to integer 
multiples of $\varepsilon$ in a way that ensures that the rounded auction is near-optimal also for the original distribution $F_1 \times \cdots \times F_n$.
The analysis uses basic tools, is performed in its entirety in value-space (rather than quantile-space, or ``virtual-value''-space), and also simplifies the analysis of above-described previously known results for special cases.
We now give a high-level review of our analysis for the basic case of a single-item auction. In \cref{single-parameter}, we extend our analysis to a variety of single-parameter auction environments including matroid environments, position environments, and the public project environment, as well as to certain single-parameter auction environments where precise revenue maximization is known to be intractable, such as knapsack environments.

We start by defining what we mean by ``rounding'' an auction.
To do so, we view the domain~$[0,H]$ of the possible valuations (values) of each bidder as composed of the disjoint union of the semiopen \emph{$\varepsilon$-intervals} $\bigl[j\cdot\varepsilon,(j+1)\cdot\varepsilon\bigr)$ for $j=0,\ldots,\nicefrac{H}{\varepsilon}$.

\begin{definition}[$\varepsilon$-Coarse Auction; $\varepsilon$-Rounding]\label{intro-coarse-rounding}\leavevmode
\begin{itemize}
\item
We call an auction $\varepsilon$-\emph{coarse} if its \emph{outcome}, i.e., (item) \emph{allocation} and payment, is constant within each product of \mbox{$\varepsilon$-intervals}. That is, the outcome of an \mbox{$\varepsilon$-coarse} auction for bids $v_1,\ldots,v_n$ depends only on the indices $j_1,\ldots,j_n$ of the \mbox{$\varepsilon$-intervals} to which the bids $v_1,\ldots,v_n$ respectively belong (i.e., $j_i\cdot\varepsilon\le v_i<(j_i+1)\cdot\varepsilon$ for every~$i$).
\item
We \emph{$\varepsilon$-round} an auction $A$ to an $\varepsilon$-coarse auction $A'$ by choosing
for each $\varepsilon$-interval~$j$ of each bidder $i$ a fixed value $v^i_j \in \bigl[j\cdot\varepsilon,(j+1)\cdot\varepsilon\bigr)$ in that $\varepsilon$-interval, and defining the allocation rule of the auction $A'$ to treat all bids of bidder $i$ in that \mbox{$\varepsilon$-interval} as the auction $A$ treats $v^i_j$. In other words, $A'$ is the unique $\varepsilon$-coarse auction whose allocation on every tuple $(v^1_{j_1},\ldots,v^n_{j_n})$ of these chosen values  is identical to that of~$A$. The payments of $A'$ are defined to be the unique payments that make it incentive compatible, and are thus generally \emph{not} identical to those of $A$.
\end{itemize}
\end{definition}

Our notion of rounding auctions allows for many different ways of $\varepsilon$-rounding any given auction, depending on the choice of the values $v^i_j$.
One might hope that all of these roundings achieve approximately the same revenue as the original auction, but that
is not the case.\footnote{While it is true that the price paid by a bidder i for winning the item in a rounded auction is indeed approximately the same whether or not we round all bids of that bidder, there are two complications:  the minor one is that bids of bidder $i$ that used to win the item in the original auction
may cease to do so in the rounded auction; the major one is that in the original auction, tiny changes in the bid of bidder~$i$ may result in arbitrary changes in the
payment of another bidder $j$.}
Our main \lcnamecref{intro-round} shows that \emph{some} $\varepsilon$-rounding of any given auction does achieve approximately the revenue of the nonrounded auction.  Moreover, for an auction that is optimal for the product of uniform distributions over finite sets of (sampled) values,  finding such an $\varepsilon$-rounding can be done in polynomial time.

\begin{lemma}[See also \cref{exists,efficient}]\label{intro-round}\leavevmode
\begin{parts}
\item\label{intro-round-exists}
	For every product distribution $F_1 \times \cdots \times F_n$ and auction $A$,
	there exists an $\varepsilon$-coarse auction $A'$ that is an $\varepsilon$-rounding of $A$ and whose revenue from $F_1 \times \cdots \times F_n$ is less than an additive $\varepsilon$ smaller than that of $A$.
\item\label{intro-round-efficient}
	There exists a deterministic algorithm that runs in time $poly(H,n,\nicefrac{1}{\varepsilon},t)$, takes as input\ \ i) a product distribution $\hat{F}_1\times\cdots\times\hat{F}_n$ where each $\hat{F}_i$ is uniform over a multiset of $t$ (sampled) values,
	and\ \ ii) an auction $A$ that is optimal for $\hat{F}_1 \times \cdots \times \hat{F}_n$, and outputs an $\varepsilon$-coarse auction $A'$ that is an $\varepsilon$-rounding of $A$ and whose revenue from
	$\hat{F}_1 \times \cdots \times \hat{F}_n$ is less than an additive~$n\varepsilon$ smaller than that of $A$.
\end{parts}
\end{lemma}

As opposed to existing auction-rounding schemes in the literature \citep{Hart-Nisan-b,Dughmi-Li-Nisan}, the rounding in \cref{intro-round} is not ``universal'' for all distributions (i.e., the rounded auction $A'$ does not approximate the revenue of $A$ on all product distributions), but rather needs to be tailored specifically
for $F_1 \times \cdots \times F_n$.\footnote{We show that this is unavoidable by giving an example of a two-bidder auction that is optimal for 
some product of regular distributions and yet for every $\varepsilon>0$, every $\varepsilon$-coarse auction must lose $\Omega(1)$ (i.e., at least some constant independent of $\varepsilon$) revenue for some valuation profile $(v_1,v_2)$.  This
is in contrast to the case of i.i.d.\ distributions, for which we show in \cref{iid} how to $\varepsilon$-round any auction that is optimal for some $F^n=F \times\cdots\times F$
into a single $\varepsilon$-coarse auction that loses less than an additive $\varepsilon$ in revenue universally for every profile of valuations $(v_1,\ldots,v_n)$, and thus also for every distribution.
This implies a simplified proof for the recent result of \cite{ironing-in-the-dark} for the i.i.d.\ case, and highlights the difference between the general case considered here and the i.i.d.\ case.}
Furthermore, as opposed to \cite{Devanur-Huang-Psomas}, we round the optimal auction for the empirical distribution, rather than find the optimal auction for a discretized version of the empirical distribution.
Our algorithm considers a randomized rounding, where all values $j\cdot \varepsilon \le v_i < (j+1)\cdot\varepsilon$
are rounded to a fixed value $v^i_j$ that is itself randomly chosen according to the distribution~$F_i$ conditioned on being in this $\varepsilon$-interval. We show that
the \emph{expected} revenue of the resulting auction (or more accurately, distribution over auctions) is within less than an additive $\varepsilon$ of the revenue of~$A$, and so some deterministic auction within this distribution over auctions loses less than an additive $\varepsilon$ in revenue compared to $A$. When for every $i\in N$ it is the case that $F_i=\hat{F}_i$ is a uniform distribution over a multiset of $t$ (sampled) values, our algorithm deterministically and efficiently searches for a deterministic auction that loses less than an additive $n\varepsilon$, within the above-defined randomized distribution over possible roundings.

Once we have \cref{intro-round}, the algorithm for empirical revenue maximization is simple: Let $A$ be \citeauthor{Myerson}'s optimal auction for the empirically sampled product distribution $\hat{F}_1 \times \cdots \times \hat{F}_n$ 
(which is readily computed from the samples as spelled out in \citealp{Elkind}), and output the $\varepsilon$-rounded approximately optimal auction $A'$ produced by the rounding algorithm for $A$, tailored to lose less than $n\varepsilon$ revenue on $\hat{F}_1 \times \cdots \times \hat{F}_n$.
When the \mbox{$\varepsilon$-rounding} algorithm is applied
to \citeauthor{Myerson}'s (ironed-virtual-welfare maximizing) optimal auction~$A$, the resulting \mbox{$\varepsilon$-coarse} auction $A'$ turns out to be of the 
following simple form:

\begin{definition}[$(H,\varepsilon)$-simple auction]
	We call an auction $(H,\varepsilon)$-simple if there exists a sequence $P$ of distinct pairs of the form $(i,j)$ where $i\in\{1,\ldots,n\}$ and $j\in\bigl\{0,\ldots,\lfloor\nicefrac{H}{\varepsilon}\rfloor\bigr\}$, such
	that the winner is always the first bidder $i$ in the order $P$ that has $v_i \ge j \cdot \varepsilon$, and the winner's payment is always the threshold value required to win.\footnote{
	This may be viewed as an $\varepsilon$-discretized variant of the ``leveled auctions'' of \cite{Morgenstern-Roughgarden}, albeit with more freedom in tie breaking; this finite discretization is exactly what allows us to provide an efficient algorithm rather than an information-theoretic result. The way we bound the number of $(H,\varepsilon)$-simple auctions is reminiscent of the way in which \cite{Devanur-Huang-Psomas} bound the number of optimal auctions on a finite valuation space.}
\end{definition}

As $(H,\varepsilon)$-simple auctions have a concise polynomial-length description (the sequence $P$), basic learning-theory intuition implies that overfitting can be ruled out.  Specifically,
we apply an elegant concentration inequality due to \cite{Babichenko-Barman-Peretz} \citep[see also][]{Devanur-Huang-Psomas},
which shows that the revenue of any fixed auction from the product of the empirically sampled distributions 
well approximates its revenue from the product of the true distributions. We then conclude by taking a union bound over
the small number of $(H,\varepsilon)$-simple auctions to obtain our main \lcnamecref{intro-empirical}:

\begin{theorem}[See also \cref{empirical}]\label{intro-empirical}
	There exists a deterministic algorithm that runs in time $poly(H,n,\nicefrac{1}{\varepsilon},\log\nicefrac{1}{\delta})$, takes as input $t=poly(H,n,\nicefrac{1}{\varepsilon},\log\nicefrac{1}{\delta})$ random samples from unknown distributions
	$F_1,\ldots,F_n$ supported on $[0,H]$ and, with probability at least~$1\!-\!\delta$, outputs a description of
	an $n$-bidder single-item auction that maximizes the revenue from $F_1 \times \cdots \times F_n$ up to less than an additive $\varepsilon$. The produced auction is an $(\frac{\varepsilon}{n+2},H)$-simple 
	$\frac{\varepsilon}{n+2}$-rounding of the auction that maximizes revenue from the empirical product of the uniform distributions over the samples.
\end{theorem}

As \citeauthor{Myerson}'s characterization of optimal auctions applies to general \emph{single-parameter} auction environments,
we are able to generalize \cref{intro-empirical} to a wide variety of computationally tractable single-parameter environments by employing similar, yet somewhat more delicate, analysis.
First, our main ingredient, \crefpart{intro-round}{exists}, readily generalizes to arbitrary single-parameter environments, with the error term multiplied by the expected ``number of winners''.
Nevertheless, it is not possible to generalize \crefpart{intro-round}{efficient} since the underlying algorithm requires the ability to compute the expected overall revenue of an auction, which may be computationally intractable 
even for computationally tractable environments. Therefore, we instead take a different approach to prove an analogue of \crefpart{intro-round}{efficient}, by applying random sampling and then derandomizing the process using random bits obtained from the order in which (polynomially many) samples are drawn. Finally, while bounding the number of resulting $\varepsilon$-coarse auctions is easy in some environments, a more sophisticated argument based on Cramer's rule is required more generally.\footnote{We are not aware of the use of any similar argument in the literature on mechanism design, and hope that this type of argument may find additional uses in similar contexts in the future.}
Putting all of these together, we derive the following generalization of \cref{intro-empirical}, which also pushes the boundary of the set of environments that previous papers handled even for the special cases of regular or i.i.d.\ distributions:

\begin{sloppypar}
\begin{theorem}[See also \cref{empirical-single-parameter,empirical-single-parameter-lt}]\label{intro-empirical-single-parameter}
	For a class of single-parameter auction environments including computationally tractable deterministic environments (such as matroid environments and the public project environment) and position environments,
	there exists a deterministic algorithm that runs in time $poly(H,n,\nicefrac{1}{\varepsilon},\log\nicefrac{1}{\delta})$, takes as input $t=poly(H,n,\nicefrac{1}{\varepsilon},\log\nicefrac{1}{\delta})$ random samples from unknown distributions
	$F_1,\ldots,F_n$ supported on $[0,H]$ and, with probability at least $1\!-\!\delta$, outputs a description of
	an $n$-bidder auction for the given auction environment that maximizes the revenue from $F_1 \times \cdots \times F_n$ up to less than an additive~$\varepsilon$.
\end{theorem}
\end{sloppypar}

Finally, for single-parameter auction environments where precise revenue maximization is known to be intractable, but where efficient approximate revenue maximization up to some multiplicative factor $C$ is possible, we prove, under certain assumptions on this \mbox{up-to-$C$} maximization algorithm (that are satisfied, e.g., in the case of knapsack environments), a generalization of \cref{empirical-single-parameter}, providing a deterministic polynomial-time algorithm that learns from polynomially many samples a \emph{tractable} auction that with high probability approximates the maximum revenue from $F_1 \times \cdots \times F_n$ up to the same multiplicative factor of $C$, plus less than an additive $\varepsilon$. (See \cref{empirical-approx-single-parameter}.) As with \cref{intro-empirical-single-parameter}, this result also pushes the boundary of the set of environments that previous papers handled even for the special cases of regular or i.i.d.\ distributions.

\subsection{Structure}

The remainder of this paper is structured as follows. \cref{definitions} provides definitions and some background, \cref{existence} provides the derivation of \crefpart{intro-round}{exists}, \cref{efficiency} provides the derivation of \crefpart{intro-round}{efficient}, and \cref{convergence} provides the remainder of the derivation of \cref{intro-empirical}.
In \cref{single-parameter}, we extend the above single-item analysis to more general single-parameter auction environments, and derive \cref{intro-empirical-single-parameter} and its generalization for single-parameter auction environments where precise revenue maximization is intractable.
Some examples referenced in this paper are relegated to \cref{examples}, and the proofs of some of the results stated throughout this paper are relegated to \cref{proofs}. As mentioned above, \cref{iid} presents an analogous (yet simpler) analysis for the special case of i.i.d.\ distributions, where distribution-independent rounding of optimal auctions is possible.

\section{Model, Definitions, and Background}\label{definitions}

\subsection{Auctions}

The auctions we consider are deterministic direct-revelation dominant-strategy incentive-compatible (DSIC) and ex-post individually rational (IR) single-item auctions among $n$ bidders numbered $N\eqdef\{1,2,\ldots,n\}$, each having a valuation (value) in $[0,H]$ for some $H\in\RR$ that is known to the auctioneer.  We denote 
by $r^A(v_1,\ldots,v_n)$ the revenue of auction $A$ when the bidders have values $v_1,\ldots,v_n$, and for every $i\in N$, denote by $r^A_i(v_1,\ldots,v_n)$
the revenue of auction $A$ from bidder $i$ (i.e., the payment of bidder~$i$) when the bidders have values $v_1,\ldots,v_n$. (So $r^A(v_1,\ldots,v_n)=\sum_{i} r^A_i(v_1,\ldots,v_n)$.)
We furthermore denote by $r^A_i(v_i)\eqdef\EE_{\forall j\ne i: v_j\sim F_j} r^A_i(v_1,\ldots,v_n)$ the \emph{expected} revenue of auction~$A$ from bidder $i$, when bidder $i$ has value $v_i$ and the valuations of all other bidders $j$ are independently drawn from the respective distributions $F_j$, which will be clear from context.
The overall revenue of auction $A$ from a product distribution $F_1\times\cdots\times F_n$ is denoted by
$\Rev^A(F_1 \times \cdots \times F_n)\eqdef\EE_{v_1 \sim F_1,\ldots,v_n \sim F_n} r^A(v_1,\ldots,v_n)$.

\subsection{Rounding}

We now make \cref{intro-coarse-rounding} from \cref{intro} precise, and provide some supporting notation.

\begin{definition}[$\varepsilon$-Interval]
Let $\varepsilon>0$.
An \emph{$\varepsilon$-interval} is a semiopen interval of the form $\bigl[j\cdot\varepsilon,(j+1)\cdot\varepsilon\bigr)$ for some integer $j\in\bigl\{0,\ldots,\lfloor\nicefrac{H}{\varepsilon}\rfloor\bigr\}$. We say that $j$ is the \emph{index} of this $\varepsilon$-interval.
\end{definition}

\begin{sloppypar}
\begin{definition}[$\varepsilon$-Coarse Auction]
Let $\varepsilon>0$.
An $n$-bidder auction $A$ (for valuations in $[0,H]$) is said to be \emph{$\varepsilon$-coarse}, if for every pair of valuations profiles $(v_1,\ldots,v_n),(w_1,\ldots,w_n)\in[0,H]^n$ such that for every $i\in N$ there exists $j_i$ such that $v_i,w_i\in\bigl[j_i\cdot\varepsilon,(j_i+1)\cdot\varepsilon\bigr)$, the outcome (allocation and payment) of $A$ is the same for $(v_1,\ldots,v_n)$ and for $(w_1,\ldots,w_n)$.
\end{definition}
\end{sloppypar}

We next define how auctions can be ``rounded'' into coarse ones, using a sequence of ``rounding actions.''

\begin{definition}[$\varepsilon$-Rounding]
Let $A$ be an $n$-bidder auction and let $\varepsilon>0$.
\begin{itemize}
\item An \emph{$\varepsilon$-rounding action} is a triplet $(i,j,v^i_j)$ where $i\in N$ is a bidder, $j$ is an index of an $\varepsilon$-interval,
and $v^i_j \in \bigl[j\cdot\varepsilon,(j+1)\cdot\varepsilon\bigr)$ is  a value in that $\varepsilon$-interval.
Applying the $\varepsilon$-rounding action $(i,j,v^i_j)$  to the auction $A$ yields an auction $A'$ having the following allocation rule: given bids $v_1,\ldots,v_n$, 
if $v_i \in \bigl[j\cdot\varepsilon,(j+1)\cdot\varepsilon\bigr)$, then change the bid of $v_i$ to $v^i_j$ (for other bidders, or for the $i$th bidder
if her value is not in this $\varepsilon$-interval, keep the original bid unchanged), and then run the allocation rule of $A$ on
the (possibly) updated values. The payment rule of $A'$ is \emph{not} directly taken from~$A$, but is rather what is needed to
ensure truthfulness: each winning bidder pays her \emph{minimal winning bid} (i.e., the infimum of all bids that would have still allowed her to win when the bids of all other bidders are unchanged).
\item An \emph{$\varepsilon$-rounding rule} is a collection of $\varepsilon$-rounding actions $(i,j,v^i_j)$, one action for each pair of bidder $i\in N$ and index $j\in\bigl\{0,\ldots,\lfloor\nicefrac{H}{\varepsilon}\rfloor\bigr\}$ of an $\varepsilon$-interval.
Applying an $\varepsilon$-rounding rule to the auction $A$ yields an auction $A'$ obtained by applying all \mbox{$\varepsilon$-rounding} actions in the rule, in arbitrary order, to $A$.  I.e., given bids $v_1,\ldots,v_n$, for every bidder $i$ the bid $v_i$ is changed to $v^i_j$ where $j$ is the index of the $\varepsilon$-interval that contains the original bid $v_i$. Each winning bidder pays her minimal winning bid. Such an auction $A'$ is called an \emph{$\varepsilon$-rounding} of $A$.
\end{itemize}
\end{definition}

\begin{remark}
Any $\varepsilon$-rounding of any auction is $\varepsilon$-coarse.
\end{remark}

\subsection{Optimal Auctions}

Our analysis makes very weak use of \citeauthor{Myerson}'s (\citeyear{Myerson}) characterization of optimal single-item auctions, which we therefore only present to the extent required by our analysis. (We emphasize this weak use of \citeauthor{Myerson}'s characterization of optimal single-item auctions by using the nonstandard term ``Myersonian Auction.'')

\begin{definition}[Myersonian Auction, \citealp{Myerson}]\label{myersonian}
An $n$-bidder \emph{Myersonian auction} (for valuations in $[0,H]$) is a tuple $(\phi_i)_{i\in N}$, where for every $i\in N$,\ \ $\phi_i:[0,H]\rightarrow\mathbb{R}$ is a nondecreasing function called the \emph{ironed virtual valuation} of bidder $i$. In this auction, there is a winner unless $\phi_i(v_i)<0$ for all $i\in N$, and the winner is the bidder with lowest index among those whose bid $v_i$ maximizes $\phi_i(v_i)$; the winner pays her minimal winning bid.\footnote{This auction can also be made symmetric, by choosing the winner uniformly at random between all bidders whose bid $v_i$ maximizes $\phi_i(v_i)$, and adapting the payments accordingly. The revenue is the same either way. In this paper we use lexicographic ordering for simplicity.}
\end{definition}

\cite{Myerson} proved that for every continuous product distribution, there exists a Myersonian auction that obtains the optimal revenue. \cite{Elkind} showed the same for discrete product distributions, giving an efficient algorithm for computing this optimal auction.

\begin{theorem}[\citealp{Myerson}]\label{myerson}
For every product distribution $F=F_1\times\cdots\times F_n$, there exists a Myersonian auction $(\phi_i)_{i\in N}$, denoted $\OPT(F)$, that achieves maximum revenue from $F$ among all possible auctions. Moreover, for every $i\in N$, the ironed virtual valuation $\phi_i$ depends only on $F_i$.
\end{theorem}

\begin{theorem}[\citealp{Elkind}]\label{elkind}
Let $t\in\mathbb{N}$. There exists an algorithm that runs in time $\poly(t)$, such that given
a discrete distribution $\hat{F}_i$ with support of size at most $t$, outputs a nondecreasing function $\phi_i:\supp{\hat{F}_i}\rightarrow\mathbb{R}$ (so, $\phi_i$ is a nondecreasing sequence of at most~$t$ real numbers), such that for every product $\hat{F}=\hat{F}_1\times\cdots\times\hat{F}_n$ of discrete distributions each having support of size at most $t$, the Myersonian auction $(\phi_i)_{i\in N}$ (where $\phi_i$ is the output of the algorithm given $\hat{F}_i$)
achieves maximum revenue from $\hat{F}$ among all possible auctions. 
\end{theorem}

\begin{remark}\label{extend-beyond-support}
The algorithm of \cref{elkind} outputs an ironed virtual valuation $\phi_i$ whose domain is $\supp\hat{F}_i$ rather than~$[0,H]$.
To be completely formal and avoid any ambiguities with regard to the definition of, e.g., the minimal winning bid of any bidder (and later on, with regard to the revenue of this auction from distributions with support larger than that of $\hat{F}$), we emphasize that whenever we consider $(\phi_i)_{i\in N}$ as a Myersonian auction, we do so by interpreting each~$\phi_i$ as specifying only the ``steps/jumps'' in the right-continuous step function that constitutes the ironed virtual valuation. That is, we implicitly extend the definition of each such~$\phi_i$ over all of~$[0,H]$ by defining ${\phi_i(v_i)\eqdef\Max\bigl\{\phi_i(w_i)~\big|~ w_i\le v_i \And \mbox{$\phi_i(w_i)$ is defined}\bigr\}}$ for every~$v_i$, where this maximum is defined to be~$-\infty$ (or any sufficiently small number) if it is taken over the empty set (i.e., this is the level of the leftmost ``plateau'' of $\phi_i$). In particular, we note that the minimal winning bid of any bidder $i$ in the Myersonian auction defined in \cref{elkind} is always in $\supp\hat{F}_i$.
\end{remark}

\section{Rounding Arbitrary Auctions}\label{existence}

In this \lcnamecref{existence}, we derive \crefpart{intro-round}{exists} from \cref{intro}. In fact, we prove a stronger result for arbitrary (not necessarily optimal) auctions. Fix an arbitrary auction $A$, fix $\varepsilon>0$, and fix a product distribution over valuation profiles $F_1\times\cdots\times F_n\in\Delta\bigl([0,H]\bigr)^n$. In this \lcnamecref{existence}, we show that there exists an $\varepsilon$-rounding of $A$ whose revenue from $F_1\times\cdots\times F_n$ is less than an additive $\varepsilon$ smaller than that of $A$.

We note that in \cref{triangle} in \cref{examples}, we show that the desired $\varepsilon$-rounding~$A'$ of $A$ has to be constructed specifically for the target distribution $F_1\times\cdots\times F_n$, i.e., that in some scenarios, every $\varepsilon$-rounding of $A$ must lose $\Omega(1)$ (i.e., at least some constant independent of~$\varepsilon$) revenue for some product distribution. We show this by showing that in such scenarios, every $\varepsilon$-rounding of $A$ must lose $\Omega(1)$ revenue for some valuation profile. We show this even for the case where $A$ is restricted to be the optimal auction for some $F_1\times\cdots\times F_n$, even when $F_1,\ldots,F_n$ are regular\footnote{Regularity is an assumption on distributions over valuations, which simplifies many analyses. We do not elaborate on its definition, as we do not require it for our analysis beyond remarking that all distributions in \cref{triangle} are regular.} distributions, even when the ``rounded'' auction $A'$ may be any $\varepsilon$-coarse auction (whether or not an $\varepsilon$-rounding of~$A$), and already for $n\!=\!2$~bidders.
In contrast, in \cref{iid} we show that for the special case in which $A$ is restricted to be the optimal auction for some i.i.d.\ product distribution $F^n=F\times\cdots\times F$, it is always possible to construct an $\varepsilon$-rounding of $A$ that loses less than an additive $\varepsilon$ over any valuation profile, and thus over any distribution, by simply rounding-down the parameters representing the optimal auction.\footnote{The construction underlying \cref{triangle} in \cref{examples} also demonstrates that for non-i.i.d.\ product distributions, simply rounding-down all bids and applying the allocation rule of the optimal auction to the rounded-down bids (while adapting the payments to ensure truthfulness, of course) may lose $\Omega(1)$ in overall revenue. See \cref{triangle-round-down} in \cref{examples} for more details.}
As shown in \cref{iid}, this implies a simplified proof for the recent result of \cite{ironing-in-the-dark} for the i.i.d.\ case, and highlights the difference between the general case considered here and the i.i.d.\ case.

Our proof strategy for \crefpart{intro-round}{exists} is a probabilistic one: Given an auction $A$ and a target distribution $F=F_1\times\cdots\times F_n$, we will construct a distribution over $\varepsilon$-roundings of $A$ that \emph{in expectation} approximates the revenue of $A$ from $F$, and hence deduce that at least one (deterministic) auction in the support of this distribution does not lose much revenue compared to $A$. We begin by defining this distribution over auctions, which rounds the bids in each $\varepsilon$-interval to a value that is randomly picked according to the restriction of the target distribution to that $\varepsilon$-interval.

\begin{definition}[Conditional Distribution on an $\varepsilon$-Interval]
For a distribution $F\in\Delta\bigl([0,H]\bigr)$ and an index $j$ of an $\varepsilon$-interval, we denote the conditional distribution of $v\sim F$ restricted to the $\varepsilon$-interval $\bigl[j\cdot\varepsilon,(j+1)\cdot\varepsilon\bigr)$ by $F|_j$.\footnote{If this conditional distribution is ill defined, i.e., if $F\bigl([j\cdot\varepsilon,(j+1)\cdot\varepsilon)\bigr)=0$, then we say that $F|_j$ is not defined.}
\end{definition}

\begin{definition}[Randomized $\varepsilon$-Rounding]
Let $\varepsilon>0$.
\begin{itemize}
\item
Let $F_i\in\Delta\bigl([0,H]\bigr)$. The \emph{$F_i$-randomized $\varepsilon$-rounding action} on bidder $i$'s $\varepsilon$-interval~$j$ is a distribution over $\varepsilon$-rounding actions $(i,j,v^i_j)$ obtained by randomly choosing $v^i_j \in \bigl[j\cdot\varepsilon,(j+1)\cdot\varepsilon\bigr)$ according to the conditional distribution $F_i|_j$ of $F_i$ restricted to this $\varepsilon$-interval.\footnote{If $F_i|_j$ is not defined, then we pick $v^i_j\in\bigl[j\cdot\varepsilon,(j+1)\cdot\varepsilon\bigr)$ arbitrarily, say, $v^i_j\eqdef j\cdot\varepsilon$.}
Note that applying the $F_i$-randomized $\varepsilon$-rounding action to an auction $A$ yields a distribution over deterministic auctions~$A'$.
\item
Let $F=F_1\times\cdots\times F_n\in\Delta\bigl([0,H]\bigr)^n$.
The \emph{$F$-randomized $\varepsilon$-rounding rule} is given by the collection of $F_i$-randomized $\varepsilon$-rounding actions for all bidders $i$ and all \mbox{$\varepsilon$-intervals}~$j$.
Note that applying the $F$-randomized $\varepsilon$-rounding action to an auction~$A$ yields a distribution over deterministic $\varepsilon$-coarse auctions~$A'$.
This $A'$ is called the \emph{$F$-randomized $\varepsilon$-rounding} of $A$.
\end{itemize}
\end{definition}

We start by analyzing the impact on the revenue from each bidder following the application of a single $F_i$-randomized $\varepsilon$-rounding action to some given auction. This is done in the following \lcnamecref{one-by-one}, whose proof, once this \lcnamecref{one-by-one} is stated, is quite straightforward (but is nonetheless spelled out in \cref{proofs}, to which all omitted proofs are relegated).

\begin{lemma}\label{one-by-one}
Let $F_1\times\cdots\times F_n \in \Delta\bigl([0,H]\bigr)^n$, let $A$ be an $n$-bidder auction, let $\varepsilon>0$, and let $i\in N$.
Let $j$ be an index of an $\varepsilon$-interval, and $A'$ be the result of applying the $F_i$-randomized $\varepsilon$-rounding action on bidder $i$'s interval $j$ to $A$.
Let $v_{-i}\in[0,H]^{N\setminus\{i\}}$ be a profile of valuations for all bidders other than $i$.
\begin{parts}
\item\label{one-by-one-others}
Let $i'\in N\setminus\{i\}$.
\begin{parts}
\item\label{one-by-one-others-always}
For every $v_i\notin\bigl[j\cdot\varepsilon,(j+1)\cdot\varepsilon\bigr)$, it is surely\footnote{I.e., for every realization in the support of $A'$.} the case that
\[r_{i'}^{A'}(v_1,\ldots,v_n)=r_{i'}^A(v_1,\ldots,v_n).\]
\item\label{one-by-one-others-amortized}
If $F_i|_j$ is defined, then
\[\EE_{v_i\sim F_i|_j}\left[\EE_{A'}r_{i'}^{A'}(v_1,\ldots,v_n)\right]=\EE_{v_i\sim F_i|_j}\left[r_{i'}^A(v_1,\ldots,v_n)\right].\]
\end{parts}
\item\label{one-by-one-i}
Let $w_i$ be the minimal winning bid of $i$ in $A$ when all other bidders bid $v_{-i}$.\footnote{The fact that $A$ is a direct-revelation DSIC and ex-post IR auction implies that a winning bidder pays her minimal winning bid, which is determined by $v_{-i}$.}
\begin{parts}
\item\label{one-by-one-i-notin}
If $w_i \notin\bigl(j\cdot\varepsilon,(j+1)\cdot\varepsilon\bigr)$,\footnote{\label{open-interval}This is indeed an \emph{open} interval. There is no typo here.} then for every $v_i\in[0,H]$, it is surely the case that
\[r_i^{A'}(v_1,\ldots,v_n)=r_i^A(v_1,\ldots,v_n).\]
\item\label{one-by-one-i-in}
Assume that $w_i\in \bigl(j\cdot\varepsilon,(j+1)\cdot\varepsilon\bigr)$.\textsuperscript{\rm{\labelcref{open-interval}}}
\begin{enumerate}
\item\label{one-by-one-i-in-always}
For every $v_i\notin\bigl[j\cdot\varepsilon,(j+1)\cdot\varepsilon\bigr)$, it is surely the case that $v_i$ wins against~$v_{-i}$ in $A'$ if and only if $v_i$~wins against $v_{-i}$ in $A$.
\item\label{one-by-one-i-in-amortized}
If $F_i|_j$ is defined, then
the winning probability of a bid~$v_i\sim F_i|_j$ against~$v_{-i}$ is the same in $A'$ and in $A$ (where in the former, the probability is taken over the randomness of both $v_i$ and $A'$, and in the latter --- over the randomness of $v_i$).
\item\label{one-by-one-i-in-payment}
Let $w'_i$ be the minimal winning bid of $i$ in $A'$ when all other bidders bid~$v_{-i}$ (so $w'_i$ is a random variable). It is surely the case that
\[\left|w'_i-w_i\right|<\varepsilon.\]
\end{enumerate}
\end{parts}
\end{parts}
\end{lemma}

\cref{one-by-one} implies the following \lcnamecref{randomized-rounding}, which we believe to be of independent interest.

\begin{theorem}\label{randomized-rounding}
For every $F=F_1\times\cdots\times F_n \in \Delta\bigl([0,H]\bigr)^n$, for every $n$-bidder auction $A$, and for every $\varepsilon>0$, it is the case that
\[\left|\EE_{A'} \Rev^{A'}(F) - \Rev^A(F) \right|<p\cdot\varepsilon,\]
where $A'$ is the $F$-randomized $\varepsilon$-rounding of $A$, and where $p$ is the probability that some bidder wins in $A$ when the profile of bids is distributed according to $F$.\footnote{In fact, a slightly stronger statement also holds, where $p$ is replaced with the probability that some bidder wins in $A$ \emph{and pays a price that is not an integer multiple of $\varepsilon$} when the profile of bids is distributed according to~$F$. See the proof of \cref{randomized-rounding} for details. This implies, in a sense, that the revenue loss due to $\varepsilon$-rounding is smaller if $A$ already behaves similarly to an $\varepsilon$-coarse auction for a set of valuation profiles that has positive probability.}
\end{theorem}

\begin{sloppypar}
It is useful, though, not to merely consider \cref{randomized-rounding} as a consequence of \cref{one-by-one}, but to consider ``intermediate results'' between these two, both for instructive purposes (indeed, the deduction of \cref{randomized-rounding} from \cref{one-by-one} involves quite a few summations), but also as these results are of independent value. For example, \cref{randomized-rounding} may be obtained by \cref{one-by-one} by first considering the overall effect of a single randomized rounding action on the overall expected revenue.
\end{sloppypar}

\begin{lemma}\label{coarse-pij}
Let $F=F_1\times\cdots\times F_n \in \Delta\bigl([0,H]\bigr)^n$, let $A$ be an $n$-bidder auction, let $\varepsilon>0$, let $i\in N$, and let $j$ be an index of an $\varepsilon$-interval.
Then,
\[\left|\EE_{A'} \Rev^{A'} - \Rev^A\right| < \varepsilon\cdot p^i_j,\]
where $A'$ is the result of applying the $F_i$-randomized $\varepsilon$-rounding action on bidder $i$'s \mbox{$\varepsilon$-interval} $j$ to $A$,
and where $p^i_j$ is the probability that $i$ both wins and pays a price in $\bigl(j\cdot\varepsilon,(j+1)\cdot\varepsilon\bigr)$ in $A$ when the profile of bids is distributed according to $F$.
\end{lemma}

\cref{coarse-pij} will also be useful in \cref{efficiency}. Another way to obtain \cref{randomized-rounding} from \cref{one-by-one} is to first consider the impact of the randomized rounding rule (i.e., all randomized rounding actions together), albeit only on the revenue from a single bidder $i$ and only when her bid is restricted to some $\varepsilon$-interval.

\begin{lemma}\label{coarse-amortized}
Let $F=F_1\times\cdots\times F_n \in \Delta\bigl([0,H]\bigr)^n$, let $A$ be an $n$-bidder auction, and let $\varepsilon>0$.
Let $A'$ be the $F$-randomized $\varepsilon$-rounding of $A$.
For every $i\in N$ and for every index~$j$ of an $\varepsilon$-interval such that $F_i|_j$ is defined, it is the case that
\[\left|\EE_{v_i\sim F_i|_j}\left[\EE_{A'} r_i^{A'}(v_i) - r_i^A(v_i)\right]\right| < \varepsilon\cdot p,\]
where $p$ is the probability that $i$ wins in $A$ when bidding $v_i\sim F_i|_j$, when the remaining bids are distributed according to $F_{-i}$.\footnote{In fact, a slightly stronger statement also holds, where $p$ is replaced with the probability that $i$ wins in~$A$ \emph{and pays a price that is not an integer multiple of $\varepsilon$} when bidding $v_i\sim F_i|_j$, when the remaining bids are distributed according to $F_{-i}$. See the proof of \cref{coarse-amortized} for details.}
\end{lemma}

Recall that we show in \cref{triangle} that an $\varepsilon$-coarse, deterministic valuation-by-valuation approximation for $A$ cannot exist for a product of general (even regular) non-i.i.d.\ distributions. \cref{coarse-amortized} precisely reveals the gap between the approximation developed in this \lcnamecref{existence} and this (unattainable) goal: The approximation developed in this \lcnamecref{existence} provides a randomized (rather than deterministic) approximation that is amortized over each $\varepsilon$-interval of valuations (rather than attained on each valuation profile separately). 
I.e., the difference between the revenues of $A$ and of its $F$-randomized $\varepsilon$-rounding from any bidder $i$, in expectation over $v_i\sim F_i|_j$ for any index $j$ of an $\varepsilon$-interval, is small.

Finally, by \cref{randomized-rounding}, at least one deterministic auction in the support of $A'$ loses less than an additive $p\cdot\varepsilon\le\varepsilon$ in (overall) revenue compared to $A$, as formalized in \cref{exists} below, which is a restatement of \crefpart{intro-round}{exists}.

\begin{proposition}[Restatement of \crefpart{intro-round}{exists}]\label{exists}
For every product distribution $F=F_1\times\cdots\times F_n \in \Delta\bigl([0,H]\bigr)^n$ and for every auction $A$, 
there exists a (deterministic, $\varepsilon$-coarse) $\varepsilon$-rounding $A'$ of $A$ whose revenue from $F$ is less than an additive $\varepsilon$ smaller than that of $A$.
\end{proposition}

\begin{proof}
Let $A'$ be the $F$-randomized $\varepsilon$-rounding of~$A$.
By \cref{randomized-rounding}, we have that
$\EE_{A'} \Rev^{A'}(F) > \Rev^A(F) - \varepsilon$. Therefore, at least one of the deterministic ($\varepsilon$-coarse, $\varepsilon$-roundings of $A$) realizations of $A'$ has revenue greater than $\Rev^A(F) - \varepsilon$.
\end{proof}

We note that in \cref{tight} in \cref{examples}, we show that the bound of ``less than~$\varepsilon$'' on the revenue difference in \crefpart{intro-round}{exists}/\cref{exists} (and hence also in \cref{randomized-rounding}, etc.)\ cannot be unconditionally tightened any further, even when $A'$ may be any $\varepsilon$-coarse auction (whether or not an $\varepsilon$-rounding of~$A$), and already for $n\!=\!1$~bidder.

\section{Efficiently Rounding Myersonian Auctions for\texorpdfstring{\\}{ }Empirical Distributions}\label{efficiency}

In the previous \lcnamecref{existence}, we have seen that for any target distribution $F=F_1\times\cdots\times F_n$, the revenue-maximizing auction~$A$ can be $\varepsilon$-rounded into an $\varepsilon$-coarse auction $A'$ that approximately maximizes the revenue from $F$. In this \lcnamecref{efficiency}, we derive \crefpart{intro-round}{efficient} from \cref{intro}, showing how to deterministically and efficiently find a ``good enough'' $\varepsilon$-rounding of~$A$ when for every $i\in N$, it is the case that $F_i=\hat{F}_i$ is a uniform distribution over a multiset of $t$ (sampled) values. (Recall that in such a scenario, the revenue-maximizing auction~$A$ can be computed in polynomial time using \cref{elkind}.) 
First, we note that if $A$ is a Myersonian auction, then any rounding of $A$ is also Myersonian, and its ironed virtual valuations can be efficiently computed from those of $A$, given the rounding actions.

\begin{remark}
Let $A=(\phi_i)_{i\in N}$ be an $n$-bidder Myersonian auction and let $\varepsilon>0$.
\begin{itemize}
\item
Applying an $\varepsilon$-rounding action $(i,j,v^i_j)$ to $A$ yields the Myersonian auction $A'$ obtained by modifying $\phi_i$ to have the constant value $\phi_i(v^i_j)$ on the $\varepsilon$-interval $\bigl[j\cdot\varepsilon,(j+1)\cdot\varepsilon\bigr)$.
\item
Applying an $\varepsilon$-rounding rule to $A$ (recall that an $\varepsilon$-rounding rule is a collection of $\varepsilon$-rounding actions $(i,j,v^i_j)$, one action for each pair of bidder $i$ and index $j$ of an $\varepsilon$-interval) yields the Myersonian auction $A'$ obtained by modifying, for each bidder $i\in N$, the function $\phi_i$ to be constant on each $\varepsilon$-interval, having the constant value $\phi_i(v^i_j)$ on $\bigl[j\cdot\varepsilon,(j+1)\cdot\varepsilon\bigr)$ for each $j$.
\item
In particular, any $\varepsilon$-rounding of any Myersonian auction is also Myersonian (and $\varepsilon$-coarse).
\end{itemize}
\end{remark}

So, how do we find a ``good enough'' rounding of the optimal Myersonian auction~$A$?
One option, which would lead to a \emph{randomized} polynomial-time algorithm for finding a suitable rounding with high probability, is to simply draw a polynomial number of deterministic roundings from the $F$-randomized $\varepsilon$-rounding of $A$, and pick the one that gives highest approximate (empirical) revenue among all drawn roundings, where the (empirical) revenue from a given rounding can be estimated by drawing a polynomial number of valuation profiles from $F$ and averaging the (empirical) revenue from the given rounding over them. We next show that when for every $i\in N$, it is the case that $F_i=\hat{F}_i$ is a uniform distribution over a multiset of $t$ (sampled) values, this process can be derandomized to yield an efficient deterministic algorithm that finds a suitable rounding.\footnote{Another, in some sense less elegant, approach to derandomizing this algorithm would be to carefully use the randomness inherent in the order of the $t$ given samples. We indeed use such an approach later when analyzing more general single-parameter auction environments in \cref{efficiency-single-parameter}.}

We start with the question of deterministically computing the (empirical) revenue $\Rev^{A'}(\hat{F})$ of a given auction $A'$, where $\hat{F}=\hat{F}_1\times\cdots\times\hat{F}_n$. We note that in \cref{sharpp-complete} in \cref{examples}, we show that for general (even $\varepsilon$-coarse) auctions $A'$, this is a computationally hard problem, even when for every $i$ it is the case that $|\supp\hat{F}_i|=2$. (This indeed motivates restricting our attention to Myersonian auctions in this \lcnamecref{efficiency}.) Nonetheless, we now show that for a Myersonian auction $A'$, its empirical revenue may be precisely computed in polynomial time. The key idea behind this result is that the revenue of a Myersonian auction from a given valuation profile only depends on the two highest ironed virtual bids; therefore, the algorithm efficiently calculates the overall revenue by enumerating over all of the possibilities for these two bids. 

\begin{lemma}\label{calc-rev}
Let $n,t\in\mathbb{N}$. There exists a deterministic algorithm that runs in time $\poly(n,t)$, such that given
$n$ discrete distributions $\hat{F}_1,\ldots,\hat{F}_n$, each with support of size at most $t$, and for each $i\in N$ given a nondecreasing $\phi_i:\supp\hat{F}_i\rightarrow\mathbb{R}$ (so, $\phi_i$ is a nondecreasing sequence of at most $t$ real numbers), outputs
the revenue $\Rev^{(\phi_i)_{i\in N}}(\hat{F})$ of the Myersonian auction $(\phi_i)_{i\in N}$ over the product distribution $\hat{F}=\hat{F}_1\times\cdots\times\hat{F}_n$.
\end{lemma}

\begin{proof}
The algorithm is as follows:
\begin{enumerate}
\item
For every $i\in N$, for every $j\in N\setminus\{i\}$, and for every $v_j\in \supp\hat{F}_j$:
\begin{enumerate}
\item
If $j<i$, then let $L^i_j(v_j)=\PP_{v_i\sim\hat{F}_i}\bigl(\phi_j(v_j)\ge\phi_i(v_i)\bigr)$,
\item
If $j>i$, then let $L^i_j(v_j)=\PP_{v_i\sim\hat{F}_i}\bigl(\phi_j(v_j)>\phi_i(v_i)\bigr)$.
\item[]
(So $L^i_j(v_j)$ is the probability that the value $v_j$ of $j$ precedes the value of $i$.)
\end{enumerate}
\item
Initialize: $r\leftarrow0$.
\item
For every $i\in N$, for every $j\in N\setminus\{i\}$, and for every $v_j \in \supp\hat{F}_j$:
\begin{enumerate}
\item
If $i<j$, then let $w_i\in\hat{F}_i$ be the lowest such that $\phi_i(w_i)\ge\phi_j(v_j)$,
\item
If $i>j$, then let $w_i\in\hat{F}_i$ be the lowest such that $\phi_i(w_i)>\phi_j(v_j)$.
\item[]
(So $w_i$ is the lowest value of $i$ that precedes the value $v_j$ of $j$. Thus, $w_i$ is the payment of $i$ when both $i$ wins and $j$ is the second-highest bidder with bid $v_j$.)
\item
Update:
$r\leftarrow r + w_i\cdot\bigl(1-L^i_j(v_j)\bigr)\cdot\hat{f}_j(v_j)\cdot\smashoperator{\prod_{k\in N\setminus\{i,j\}}} L^k_j(v_j)$.
\item[]
(The added amount is the contribution to the expected revenue from valuation profiles where $i$ wins and $j$ is the second-highest bidder with bid $v_j$.)
\end{enumerate}
\item
Output: $r$.\qedhere
\end{enumerate}
\end{proof}

We are now ready to present a deterministic algorithm that, given a Myersonian auction $A=(\phi_i)_{i\in N}$ that maximizes (empirical) revenue, finds an $\varepsilon$-rounding of $A$ that maximizes (empirical) revenue up to less than an additive $n\varepsilon$. This algorithm sequentially and greedily searches for the values $v^i_j$ for $\varepsilon$-rounding $A$, calculating the interim revenues (following each rounding action) via the algorithm of \cref{calc-rev}, and bounding the revenue loss via \cref{coarse-pij}. From this algorithm, \crefpart{intro-round}{efficient} follows.

\begin{proposition}\label{efficient}
Let $n,t\in\mathbb{N}$ and let $\varepsilon>0$. There exists a deterministic algorithm that runs in time $\poly(H,n,\nicefrac{1}{\varepsilon},t)$, such that given
$n$ discrete distributions $\hat{F}_1,\ldots,\hat{F}_n$, each with support of size at most $t$, and for each $i\in N$ given a nondecreasing $\phi_i:\supp\hat{F}_i\rightarrow\mathbb{R}$ (so, $\phi_i$ is a nondecreasing sequence of at most $t$ real numbers), outputs for every $i\in N$ a function $\phi'_i$, such that $(\phi'_i)_{i\in N}$ is an $\varepsilon$-rounding of $(\phi_i)_{i\in N}$,
and such that \[\Rev^{(\phi'_i)_{i\in N}}(\hat{F})>\Rev^{(\phi_i)_{i\in N}}(\hat{F})-n\varepsilon,\]
where $\hat{F}=\hat{F}_1\times\cdots\times\hat{F}_n$.
\end{proposition}

\begin{proof}
The algorithm is as follows:
\begin{enumerate}
\item
Initialize: $\phi'_i\leftarrow\phi_i$ for every $i$.
\item
For every $j\in\bigl\{0,\ldots,\lfloor\nicefrac{H}{\varepsilon}\rfloor\bigr\}$: initialize\footnote{See \cref{extend-beyond-support}.} $\phi'_i(j\cdot\varepsilon)\gets \Max_{w_i\in \supp \hat{F}_i \cap [0,j\cdot\varepsilon]}\bigl\{\phi_i(w_i)\bigr\}$.
\item
For every $i\in N$:
\begin{enumerate}
\item
For every $j\in\bigl\{0,\ldots,\lfloor\nicefrac{H}{\varepsilon}\rfloor\bigr\}$:
\begin{enumerate}
\item
Let $\phi''_i\leftarrow\phi'_i$.
\item
For every $v^i_j\in \supp\hat{F}_i\cap\bigl[j\cdot\varepsilon,(j+1)\cdot\varepsilon\bigr)$:
\begin{enumerate}
\item
Set $\phi''_i(v_i)\leftarrow\phi'(v^i_j)$ for every $v_i\in\Bigl(\supp\hat{F}_i\cap\bigl[j\cdot\varepsilon,(j+1)\cdot\varepsilon\bigr)\Bigr)\cup\{j\cdot\varepsilon\}$.
\item[]
(So $\phi''$ is the result of applying the $\varepsilon$-rounding action $(i,j,v^i_j)$ to $\phi'$.)
\item
Compute $r_{v^i_j}\leftarrow\Rev^{(\phi_i'',\phi'_{-i})}(\hat{F})$ via the algorithm of \cref{calc-rev}.
\end{enumerate}
\item
For $v^i_j$ for which $r_{v^i_j}$ is highest:\\
Update $\phi_i'(v_i)\leftarrow\phi'(v^i_j)$ for every $v_i\in\Bigl(\supp\hat{F}_i\cap\bigl[j\cdot\varepsilon,(j+1)\cdot\varepsilon\bigr)\Bigr)\cup\{j\cdot\varepsilon\}$.
\item[]
(So $\phi'$ is updated by applying the $\varepsilon$-rounding action $(i,j,v^i_j)$ that yields the highest revenue for $\phi'$ (after updating) among all $\varepsilon$-rounding actions for bidder~$i$'s $\varepsilon$-interval $j$.)
\end{enumerate}
\end{enumerate}
\item[]
(\cref{coarse-pij} guarantees that the revenue loss is less than $\varepsilon\cdot p^i_j$ for every pair $(i,j)$, and so less than~$\varepsilon\cdot p_i\le\varepsilon$ for every $i$, where $p_i$ is the probability that $i$ wins, for a total of less than $n\varepsilon$ over all bidders. Note that $p^i_j$ does not change during our modification of $\phi_i$ at any $\varepsilon$-interval with index $k\ne j$; on the other hand, the probabilities $p_i$ may well change during our modification of $\phi_{i'}$ for $i'\ne i$, and therefore our bound on the overall revenue loss is $n\varepsilon$ rather than $\varepsilon$ as in the case of randomized rounding studied in \cref{existence}.)
\item
Output: $(\phi'_i)_{i\in N}$.\qedhere
\end{enumerate}
\end{proof}

Recall that we have shown in \cref{existence} that there exists an $\varepsilon$-rounding of any given Myersonian auction for any given target distribution, that loses less than an additive $\varepsilon$ in revenue. It is unclear to us whether one can bridge the gap between this (tight) less-than-$\varepsilon$ loss and the less-than-$n\varepsilon$ loss of \cref{efficient} via an efficient deterministic algorithm.

\section{Uniform Convergence over the Set of\texorpdfstring{\\}{ }Rounded Myersonian Auctions}\label{convergence}

In this \lcnamecref{convergence}, we conclude the derivation of \cref{intro-empirical} from \cref{intro}.
While most previous analyses \citep{Morgenstern-Roughgarden,ironing-in-the-dark} restricted the set of possible ``output auctions'' to an infinite set and used advanced statistical tools to obtain that all auctions in this set perform similarly well on the empirical distribution and on the true distribution, in this paper we restrict the set of possible ``output auctions'' to a finite set, for which such a uniform-convergence result may be shown via an elegant concentration inequality due to \cite{Babichenko-Barman-Peretz} \citep[see also][]{Devanur-Huang-Psomas}.

\begin{definition}[$S^n_{\varepsilon}$]
For every $\varepsilon>0$,
we denote the set of all $\varepsilon$-coarse Myersonian $n$-bidder auctions (for valuations in $[0,H]$) by $S^n_{\varepsilon}$.
\end{definition}

\begin{lemma}\label{small}
$\bigl|S^n_{\varepsilon}\bigr|\le\exp\bigl(\poly(H,n,\nicefrac{1}{\varepsilon})\bigr)$.
\end{lemma}

\begin{proof}[Proof.\footnotemark]\footnotetext{This proof is reminiscent of the way \cite{Devanur-Huang-Psomas} bound the number of optimal auctions on a finite valuation space.}
Let  $I\eqdef\bigl\{0,\ldots,\lfloor\nicefrac{H}{\varepsilon}\rfloor\bigr\}$. We claim that every $\varepsilon$-coarse Myersonian auction $A'$ (for valuations in $[0,H]$) is completely specified by a sequence $\bigl((i_1,j_1),(i_2,j_2),\ldots,(i_K,j_K)\bigr)$ of distinct pairs from $N\times I$ (so $K\le|N|\cdot|I|$). The correspondence is as follows: Given such a sequence,
the allocation rule of the corresponding auction is specified by the following algorithm:
\begin{enumerate}
\item
If $v_{i_1} \ge j_1\cdot\varepsilon$, then $i_1$ wins; else, continue.
\item
If $v_{i_2} \ge j_2\cdot\varepsilon$, then $i_2$ wins; else, continue.
\item[]
$\vdots$
\item[$K$.]
If $v_{i_K} \ge j_K\cdot\varepsilon$, then $i_K$ wins; else, continue.
\item[$K\!+\!1$.]
No bidder wins.
\item[]
The winner (if any) pays her minimal winning bid.
\end{enumerate}
Indeed, such a sequence may be constructed for any $\varepsilon$-coarse Myersonian auction $A'=(\phi'_i)_{i\in N}$ by taking all pairs $(i,j)$ such that $\phi'_i$ attains nonnegative value on the $\varepsilon$-interval with index $j$, and sorting these pairs in decreasing order of this nonnegative value, breaking ties in favor of lower $i$, and for the same $i$ in favor of higher $j$.\footnote{This results in a sequence with the additional property that for every $k\in\{1,\ldots,K\}$ and for every $j\in\bigl\{j_k+1,\ldots,\lfloor\nicefrac{H}{\varepsilon}\rfloor\bigr\}$, there exists $\ell<k$ such that $i_\ell=i_k$ and $j_\ell=j$. The above correspondence between $\varepsilon$-coarse Myersonian auctions and sequences of bidder-index pairs, when the sequences are restricted to have this additional property, is in fact one-to-one.}
\end{proof}

By \cref{small}, we obtain the following corollary of \crefpart{intro-round}{exists}/\cref{exists}, which may be of independent interest.

\begin{corollary}
For every $\varepsilon>0$, there exists a finite set of auctions $S^n_{\varepsilon}$ (of size at most $\exp\bigl(\poly(H,n,\nicefrac{1}{\varepsilon})\bigr)$), such that for every product distribution $F=F_1\times\ldots\times F_n\in\Delta\bigl([0,H]\bigr)^n$ there exists an auction $A\in S^n_{\varepsilon}$ that approximates the maximum possible revenue from $F$ up to less than an additive $\varepsilon$.
\end{corollary}

The exponential size of $S^n_{\varepsilon}$ means that if we can show for any given single auction that the number of samples required for it to perform similarly well on the empirical distribution $\hat{F}$ and on the true distribution $F$ is polynomial in $H,n,$ and $\nicefrac{1}{\varepsilon}$, but has logarithmic dependence on the desired success probability $\delta$ (as in the Chernoff-Hoeffding Inequality), then a union bound over $S^n_{\varepsilon}$ can show that a polynomial number of samples suffices to guarantee that with high probability all $\varepsilon$-coarse Myersonian auctions perform similarly well on the empirical distribution $\hat{F}$ and on the true distribution $F$.
Fortunately, an elegant concentration inequality due to \cite{Babichenko-Barman-Peretz} \citep[see also][]{Devanur-Huang-Psomas} shows precisely this (for any given single auction). The following \lcnamecref{converge-one} is a special case of this concentration inequality. (See also \cref{converge-one-iid} in \cref{iid} for an analogous yet somewhat more delicate concentration inequality for i.i.d.\ distributions, which does not follow from the analysis of \citealp{Babichenko-Barman-Peretz} or of \citealp{Devanur-Huang-Psomas}.)

\begin{proposition}[Follows from \citealp{Babichenko-Barman-Peretz}; see also \citealp{Devanur-Huang-Psomas}]\label{converge-one}
For every $\varepsilon>0$ and $\delta>0$, there
exists $t=\poly(H,n,\nicefrac{1}{\varepsilon},\log\nicefrac{1}{\delta})$ such that the following holds.
Fix $F_1,\ldots,F_n\in\Delta([0,H])$ and fix an auction $A$. Draw $t$ samples from each $F_i$, and let $\hat{F}_i$ be the empirical
uniform distribution over the $t$ samples from $F_i$.  Then, with probability at least~$1\!-\!\delta$ it is the case that
\[
\bigl|\Rev^A(\hat{F}) - \Rev^A(F)\bigr|<\varepsilon,
\]
where $F=F_1\times\cdots\times F_n$ and $\hat{F}=\hat{F}_1\times\cdots\times\hat{F}_n$.
\end{proposition}

For completeness, we provide in the \lcnamecref{proofs} a short proof of \cref{converge-one} via the Chernoff-Hoeffding Inequality and a union bound.

\begin{sloppypar}
\begin{remark}\label{converge-product-space}
By the Chernoff-Hoeffding Inequality,
when we draw $t$ tuples $(v^1_1,\ldots,v^1_n),\cdots,(v^t_1,\ldots,v^t_n)$ from $F$,
the empirical average value of the revenue of $A$ is  concentrated around its expectation $\Rev^A(F)$.  This is \emph{not} what \cref{converge-one} (and \citealp{Babichenko-Barman-Peretz} / \citealp{Devanur-Huang-Psomas}) shows, since this shows that
$\EE_{j\sim U(\{1,2,\ldots,t\})} r^A(v^j_1,\ldots,v^j_n)$ is close to $\Rev^A(F)$, while \cref{converge-one} shows that $\Rev^A(\hat{F})=\EE_{j_1 \sim U(\{1,\ldots,t\}),\cdots,j_n \sim U(\{1,\ldots,t\})} r^A(v^{j_1}_1,\ldots,v^{j_n}_n)$ is close to $\Rev^A(F)$.
\end{remark}
\end{sloppypar}

As noted above, the logarithmic dependence of $t$ on $\nicefrac{1}{\delta}$ in \cref{converge-one}, in conjunction with the exponential size of $S^n_{\varepsilon}$, allows us to use the union bound to obtain the required uniform-convergence result over $S^n_{\varepsilon}$.

\begin{lemma}\label{uniform}
For every $\varepsilon>0$ and $\delta>0$, there
exists $t=\poly(H,n,\nicefrac{1}{\varepsilon},\log\nicefrac{1}{\delta})$ such that the following holds.
Fix $F_1,\ldots,F_n\in\Delta\bigl([0,H]\bigr)$,
draw $t$ samples from each $F_i$, and let $\hat{F}_i$ be the empirical
uniform distribution over the $t$ samples from $F_i$.  
With probability at least~$1\!-\!\delta$, it is the case that
\[
\bigl|\Rev^A(F) - \Rev^A(\hat{F})\bigr|<\varepsilon
\]
holds simultaneously for every $A \in S^n_{\varepsilon}\cup\bigl\{\OPT(F)\bigr\}$, where $F=F_1\times\cdots\times F_n$ and $\hat{F}=\hat{F}_1 \times \cdots \times \hat{F}_n$.
\end{lemma}

Combining \cref{uniform,efficient,elkind}, we obtain the following \lcnamecref{empirical}, providing a polynomial-time algorithm for learning an approximately optimal auction from samples from an arbitrary unknown bounded product distribution.

\begin{theorem}[Detailed version of \cref{intro-empirical}]\label{empirical}
There exists $t=\poly(H,n,\nicefrac{1}{\varepsilon},\log\nicefrac{1}{\delta})$ such that the following holds.
Let $F_1,\ldots,F_n$ be arbitrary distributions on $[0,H]$.
For every $i\in N$, draw $t$ samples from $F_i$, and let $\hat{F}_i$ be the empirical distribution over the $t$~samples from $F_i$.
Then, with probability at least~$1\!-\!\delta$, the optimal auction for $\hat{F}=\hat{F}_1\times\cdots\times\hat{F}_n$ (which can be deterministically computed in time $\poly(t)$ via \cref{elkind}),
when \mbox{$\frac{\varepsilon}{n+2}$-rounded} via the deterministic polynomial-time algorithm of \cref{efficient}, approximates the maximum possible revenue from $F=F_1\times\cdots\times F_n$ up to less than an additive~$\varepsilon$.
\end{theorem}

\begin{proof}
Let $t$ be as in \cref{uniform}, for $S^n_{\varepsilon/(n+2)}$.
By \cref{elkind,efficient}, all that we have to show is that the $\frac{\varepsilon}{n+2}$-rounded auction $A'\in S^n_{\varepsilon/(n+2)}$ that maximizes the revenue from $\hat{F}$ up to less than an additive $n\cdot\frac{\varepsilon}{n+2}$ also maximizes the revenue from $F$ up to less than an additive $\varepsilon$.
By \cref{uniform}, by definition of $A'$, and by \cref{uniform} again, we have that
\begin{align*}
\Rev^{\OPT(F)}(F)
&< \Rev^{\OPT(F)}(\hat{F}) + \tfrac{\varepsilon}{n+2}\le \\
&\le \Rev^{\OPT(\hat{F})}(\hat{F}) + \tfrac{\varepsilon}{n+2}< \\
&< \Rev^{A'}(\hat{F}) + (n+1)\cdot\tfrac{\varepsilon}{n+2}< \\
&< \Rev^{A'}(F) + (n+2)\cdot\tfrac{\varepsilon}{n+2}=\\
&=\Rev^{A'}(F) + \varepsilon,
\end{align*}
as required.
\end{proof}

\section{Beyond Single-Item Auctions}\label{single-parameter}

In this \lcnamecref{single-parameter}, we extend the single-item analysis of the previous sections to more general single-parameter auction environments, deriving \cref{intro-empirical-single-parameter} from \cref{intro}, as well as its generalization for intractable single-parameter auction environments.
A \emph{single-parameter (auction) environment} is defined by a set of \emph{possible outcomes} $X\subseteq[0,1]^n$, where a possible outcome $(x_1,\ldots,x_n)\in X$ indicates that bidder $1$ wins a fraction $x_1\in[0,1]$, bidder $2$ wins a fraction $x_2\in[0,1]$, etc. The maximum willingness of a bidder $i$ with valuation $v_i$ to pay for a fraction $x_i$ is $x_i\cdot v_i$. The single-item environment is obtained as a special case of single-parameter environment with 
\[X_{\text{single-item}}=\bigl\{(0,\ldots,0),(1,0,\ldots,0),(0,1,0,\ldots,0),\cdots,(0,\ldots,0,1)\bigr\}.\]
Additional important examples of single-parameter environments include \emph{matroid environments} \citep[see, e.g.][]{HartlineBook}, where the possible sets of winners correspond to independent sets in some matroid $(N,I)$ over the ground set $N$ of all bidders:
\[X_{(N,I)}=\bigl\{(x_1,\ldots,x_n)\in\{0,1\}^n~\big|~\{i\in N\mid x_i=1\}\in I\bigr\};\]
and the \emph{(nonexcludable) public project environment} \citep[see, e.g.][]{HartlineBook}, where the auctioneer chooses whether or not to produce a public project from which all bidders benefit:
\[X_{\text{public-project}}=\bigl\{(0,\ldots,0),(1,\ldots,1)\bigr\}.\]
All of the above examples are special cases of \emph{deterministic environments}, in which $X\subseteq\{0,1\}^n$, i.e., each bidder either wins a fraction $1$ or loses. An important example of a single-parameter environment that is not deterministic is that of \emph{position environments} (see, e.g., \citealp{HartlineBook}; see also \emph{sponsored search auctions} in \citealp{RoughgardenBook}), where $n$ numbers $x^{(1)}\ge\cdots\ge x^{(n)}\in[0,1]$ are given, and each possible outcome corresponds to an ordering of the bidders, where the $i$th bidder in this order wins a fraction $x^{(i)}$.

\subsection{Optimal Auctions}

We once again make very weak use of \citeauthor{Myerson}'s (\citeyear{Myerson}) characterization of optimal auctions for single-parameter environments, and once again only present this characterization, which generalizes that presented in \cref{definitions}, to the extent required by our analysis.

\begin{sloppypar}
\begin{definition}[Myersonian Auction for a Single-Parameter Environment, \citealp{Myerson}]\label{myersonian-single-parameter}
An $n$-bidder \emph{single-parameter Myersonian auction} $\bigl((\phi_i)_{i\in N};\MAX_X)\bigr)$ for the single-parameter environment $X\subseteq[0,1]^n$ (and valuations in $[0,H]$) has, for every $i\in N$, a nondecreasing function $\phi_i:[0,H]\rightarrow\mathbb{R}$ called the \emph{ironed virtual valuation} of bidder $i$. In this auction,
the chosen outcome is $(x_1,\ldots,x_n)\in X$ that maximizes the \emph{ironed virtual welfare} $\sum_{i=1}^n x_i\cdot \phi_i(v_i)$, where ties between maximizing outcomes are broken lexicographically;\footnote{Once again, we use lexicographic tie breaking for simplicity.} the payments are the unique\footnote{For details, see, e.g., \cite{HartlineBook,RoughgardenBook}.} nonnegative payments that make this auction truthful for all valuation profiles in $[0,H]^n$.
\end{definition}
\end{sloppypar}

For arbitrary single-parameter environments, \cite{Myerson} proves the following generalization of \cref{myerson}.

\begin{theorem}[\citealp{Myerson}]\label{myerson-single-parameter}
For every single-parameter environment $X\subseteq[0,1]^n$ and product distribution $F=F_1\times\cdots\times F_n$, the Myersonian auction $\OPT(F;X)=\bigl((\phi_i)_{i\in N};\MAX_X\bigr)$, where for every $i\in N$ the ironed virtual valuation $\phi_i$ is as in \cref{myerson} (and so depends only on $F_i$, and can be efficiently computed for discrete $F_i$ as in \cref{elkind}), achieves maximum revenue from $F$ among all possible auctions for $X$.
\end{theorem}

The reader may verify that the Myersonian auction $\bigl((\phi_i)_{i\in N};\MAX_{X_{\text{single-item}}}\bigr)$ coincides with the Myersonian auction $(\phi_i)_{i\in N}$ defined in \cref{myersonian}.

\subsection{Rounding Arbitrary Auctions}\label{existence-single-parameter}

Analysis similar to that of \cref{existence} can be used to prove the following generalization of \cref{randomized-rounding}.

\begin{theorem}\label{randomized-rounding-single-parameter}
For every $F=F_1\times\cdots\times F_n \in \Delta\bigl([0,H]\bigr)^n$, for every $n$-bidder auction $A$ for a single-parameter environment $X\subseteq[0,1]^n$, and for every $\varepsilon>0$, it is the case that
\[\left|\EE_{A'} \Rev^{A'}(F) - \Rev^A(F) \right| < W_A\cdot\varepsilon,\]
where $A'$ is the $F$-randomized $\varepsilon$-rounding of $A$, and where $W_A\eqdef\EE_{(x_i)}\sum_{i=1}^n x_i\le n$, where the expectation is over the outcome $(x_i)_{i\in N}\in X$ of $A$ when the profile of bids is distributed according to $F$. (E.g., for a deterministic environment $X\subseteq\{0,1\}^n$,\ \ $W_A$ is the expected number of winners in the outcome of~$A$.)\footnote{In fact, as is the case for \cref{randomized-rounding}, a slightly stronger statement also holds, where $W_A$ is replaced with, roughly speaking, the expectation over the sum of the fractions of the winnings whose threshold prices are not integer multiples of $\varepsilon$. This implies here as well, in a similar sense, that the revenue loss due to $\varepsilon$-rounding is smaller if $A$ already behaves similarly to an $\varepsilon$-coarse auction for a set of valuation profiles that has positive probability.}
\end{theorem}

\noindent
By \cref{randomized-rounding-single-parameter}, we obtain the following generalization of \cref{exists}.

\begin{proposition}\label{exists-single-parameter}
For every product distribution $F=F_1\times\cdots\times F_n \in \Delta\bigl([0,H]\bigr)^n$ and for every auction $A$ for a single-parameter environment $X\subseteq[0,1]^n$, 
there exists a (deterministic, $\varepsilon$-coarse) $\varepsilon$-rounding $A'$ of $A$ whose revenue from $F$ is less than an additive $W_X\cdot\varepsilon$ smaller than that of $A$,
where $W_X\eqdef\Max_{(x_i)\in X} \sum_{i=1}^n x_i\le n$. (E.g., for a deterministic environment $X\subseteq\{0,1\}^n$,\ \ $W_X$ is the maximum number of winners in any outcome in~$X$.)
\end{proposition}

\subsection{Efficiently Rounding Myersonian Auctions for\texorpdfstring{\\}{ }Empirical Distributions}\label{efficiency-single-parameter}

While the deterministic rounding algorithm of \cref{efficient} can be adapted to certain single-parameter environments beyond the single-item environment, such as matroid and position environments, its complexity even for these two environments becomes exponential in the maximum number of winners. For this reason, we now consider a different approach for efficiently searching for a revenue-approximating deterministic rounding of a given Myersonian auction $A$, for a given target distribution $F$. Recall that in \cref{efficiency}, we sketched the following outline for a simple \emph{randomized} polynomial-time algorithm for finding, with high probability, a revenue-approximating $\varepsilon$-rounding of $A$: draw a polynomial number of deterministic roundings from the $F$-randomized $\varepsilon$-rounding of $A$, and pick the one that gives the highest approximate revenue among all drawn roundings, where the revenue from a given rounding can be estimated by drawing a polynomial number of valuation profiles from $F$ and averaging the revenue from the given rounding over them. The approach that we now present uses precisely this algorithm, however deterministically obtains the necessary randomness from a polynomial number of samples that are drawn from the true distribution of valuations, thus succeeding with high probability over the drawn samples. Since our entire deterministic algorithm from \cref{empirical}/\cref{intro-empirical} is only guaranteed to succeed with high probability over samples drawn from the true distribution of valuations, therefore utilizing a randomized algorithm to choose a suitable rounding of the empirical revenue-maximizing auction $A$ in the above-described way (deterministically obtaining the necessary randomness from samples drawn from the true distribution) does not qualitatively change the nature of our overall result: a deterministic algorithm that runs in polynomial time and, given polynomially many samples from the true distribution, outputs a (deterministic) auction that approximately maximizes revenue with high probability, where the high probability is over the randomness of the given samples.

We start by formalizing the above randomized algorithm, to obtain an analogue of \cref{efficient}, which on one hand generalizes \cref{efficient} beyond single-item auctions, but on the other hand trades determinism for randomization.

\begin{sloppypar}
\begin{definition}[Tractable Single-Parameter Environment]
We say that a single-parameter environment $X\subseteq[0,1]^n$ is \emph{tractable} if there exists a deterministic algorithm that runs in time $\poly(H,n)$, such that given the (realized) ironed virtual values $\phi_i(v_i)$ of all bidders, outputs the (lexicographically first) outcome $(x_1,\ldots,x_n)\in X$ that maximizes the ironed virtual welfare ${\sum_{i=1}^n x_i\cdot \phi_i(v_i)}$.
\end{definition}
\end{sloppypar}

\begin{lemma}\label{efficient-randomized-single-parameter}
Fix a tractable single-parameter environment $X\subseteq[0,1]^n$.
Let $n,t\in\mathbb{N}$ and let $\varepsilon>0$ and $\delta>0$. There exists a randomized algorithm that runs in time $\poly(H,n,\nicefrac{1}{\varepsilon},\log\nicefrac{1}{\delta},t)$ and uses at most $b=\poly(H,n,\nicefrac{1}{\varepsilon},\log\nicefrac{1}{\delta},\log t)$ random bits, such that given
$n$ discrete distributions $\hat{F}_1,\ldots,\hat{F}_n$, each with support of size at most $t$, and for each $i\in N$ given a nondecreasing $\phi_i:\supp\hat{F}_i\rightarrow\mathbb{R}$ (so, $\phi_i$ is a nondecreasing sequence of at most $t$ real numbers), outputs for every $i\in N$ a function $\phi'_i:\supp\hat{F}_i\rightarrow\mathbb{R}$, such that $\bigl((\phi'_i)_{i\in N};\MAX_X\bigr)$ is an $\varepsilon$-rounding of $\bigl((\phi_i)_{i\in N};\MAX_X\bigr)$, and such that with probability at least $1\!-\!\delta$ it is the case that
\[\Rev^{((\phi'_i)_{i\in N};\MAX_X)}(\hat{F})>\Rev^{((\phi_i)_{i\in N};\MAX_X)}(\hat{F})-(W_X+1)\cdot\varepsilon,\]
where $\hat{F}=\hat{F}_1\times\cdots\times\hat{F}_n$.
\end{lemma}

While for some (true) distributions $F$ (with very low entropy), it is not possible to extract $b=\poly(H,n,\nicefrac{1}{\varepsilon},\log\nicefrac{1}{\delta},\log t)$ random bits from polynomially many random samples from $F$, we now show that such distributions can be easily identified, and an approximately optimal auction can easily be directly learned for them. We are therefore able to derandomize the algorithm from \cref{efficient-randomized-single-parameter} and obtain the following counterpart for \cref{efficient} for general tractable single-parameter environments.

\begin{proposition}\label{efficient-single-parameter}
Fix a tractable single-parameter environment $X\subseteq[0,1]^n$.
Let $n,t\in\mathbb{N}$, let $\varepsilon>0$ and $\delta>0$. There exists $s=\poly(H,n,\nicefrac{1}{\varepsilon},\log\nicefrac{1}{\delta},t)\ge t$ and a
deterministic algorithm that runs in time $\poly(H,n,\nicefrac{1}{\varepsilon},\log\nicefrac{1}{\delta},t)$, such that the following holds.
Let $F_1,\ldots,F_n$ be arbitrary distributions on $[0,H]$. For every $i\in N$, draw $s$ samples from~$F_i$, and let $\hat{F}_i$ be the empirical distribution over the first $t$ of the $s$ samples from $F_i$.
The algorithm, given the $s$ samples drawn from~$F=F_1\times\cdots\times F_n$, and for each $i\in N$ given a nondecreasing $\phi_i:\supp\hat{F}_i\rightarrow\mathbb{R}$ (so, $\phi_i$ is a nondecreasing sequence of at most~$t$ real numbers), outputs a bit $e$ and for every $i\in N$ a function $\phi'_i$, where if $e=1$ then $\bigl((\phi'_i)_{i\in N};\MAX_X\bigr)$ is an \mbox{$\varepsilon$-rounding} of $\bigl((\phi_i)_{i\in N};\MAX_X\bigr)$,
such that with probability at least $1\!-\!\delta$, one of the following holds:
\begin{itemize}
\item
$e=1$ \quad and \quad $\Rev^{((\phi'_i)_{i\in N};\MAX_X)}(\hat{F})>\Rev^{((\phi_i)_{i\in N};\MAX_X)}(\hat{F})-(W_X+1)\cdot\varepsilon$,\newline
\vphantom{a}\hfill where $\hat{F}=\hat{F}_1\times\cdots\times\hat{F}_n$.
\item
$e=0$ \quad and \quad $\Rev^{((\phi'_i)_{i\in N};\MAX_X)}(F)>\Rev^{OPT(F;X)}(F)-\varepsilon$.
\end{itemize}
\end{proposition}

\subsection{Uniform Convergence over the Set of\texorpdfstring{\\}{ }Rounded Myersonian Auctions}

Recall that in the single-item environment, we obtained a uniform convergence result by noting that any $\varepsilon$-rounding of an optimal auction can be encoded by essentially only encoding the relative order of all ironed virtual values $\phi_i(v_i)$ for $i\in N$ and $v_i\in\supp \hat{F}_i$, and also of~$0$. While such an encoding suffices for single-parameter environments where the outcome that maximizes ironed virtual welfare can be found via a greedy-by-ironed-virtual-value algorithm that at each step adds the next compatible bidder with maximum nonnegative ironed virtual value~$\phi_i(v_i)$ to the set of winners, such as matroid and position environments, this encoding is unsuitable for more general environments. Indeed, already for the public project environment, knowledge of the order of all ironed virtual values $\phi_i(v_i)$ and $0$ does not suffice in order to determine for a given valuation profile $(v_1,\ldots,v_n)\in\supp \hat{F}$ whether the ironed virtual welfare $\sum_{i=1}^n\phi_i(v_i)$ for the outcome where all bidders win is greater than or less than~$0$, the ironed virtual welfare for the outcome where no bidder wins. We now show that nonetheless, for any deterministic environment (including the public project environment), if we encode every $\varepsilon$-rounding of a Myersonian auction for this environment using the relative order of the ironed virtual welfares for all possible outcomes and all valuation profiles, then the set of all possible encodings still has size at most $\exp\bigl(\poly(H,n,\nicefrac{1}{\varepsilon})\bigr)$, and so the analysis of \cref{convergence} may be used to obtain the required an appropriate uniform convergence result for arbitrary deterministic environments.

\begin{definition}[$S^X_{\varepsilon}$]
For every $\varepsilon>0$,
we denote the set of all $\varepsilon$-coarse Myersonian $n$-bidder auctions for a given single-parameter environment $X\subseteq[0,1]^n$ by $S^X_{\varepsilon}$.
\end{definition}

\begin{lemma}\label{small-single-parameter}
There exists $e\le\exp\bigl(\poly(H,n,\nicefrac{1}{\varepsilon})\bigr)$ such that $\bigl|S^X_{\varepsilon}\bigr|\le e$ for every single-parameter deterministic environment $X\subseteq\{0,1\}^n$.
\end{lemma}

\begin{proof}[Proof.\footnotemark]\footnotetext{This proof uses a technique in the spirit of the one popularly used to count weighted linear threshold circuits. We are not aware of the use of any similar argument in the literature on mechanism design, and hope that this type of argument may find additional uses in similar contexts in the future.}
Let $\bigl((\phi'_i)_{i\in N};\MAX_X\bigr)$ be an $\varepsilon$-coarse Myersonian $n$-bidder auction for $X$. Let $M\eqdef n\cdot\bigl(\lfloor\nicefrac{H}{\varepsilon}\rfloor+1\bigr)$. We will show that for every $i\in N$, there exists a function $\psi'_i:{[0,H]\rightarrow\RR}$ that is constant on each $\varepsilon$-interval and only attains values that are signed rational numbers with nominator and denominator both having absolute value at most $M!$, such that $\bigl((\psi'_i)_{i\in N};\MAX_X\bigr)=\bigl((\phi'_i)_{i\in N};\MAX_X\bigr)$, i.e., these two Myersonian auctions have the same outcome for every valuation profile. As this implies that
$\bigl|S^X_{\varepsilon}\bigr|\le(2\cdot M!^2+1)^M\le\exp\bigl(\poly(H,n,\nicefrac{1}{\varepsilon})\bigr)$, the \lcnamecref{small-single-parameter} will follow.

For every $i\in N$ and $j\in\bigl\{0,\ldots,\lfloor\nicefrac{H}{\varepsilon}\rfloor\bigr\}$, let $\phi'_{i,j}$ be the value that $\phi'_i$ attains on the \mbox{$\varepsilon$-interval} $\bigl[j\cdot\varepsilon,(j+1)\cdot\varepsilon\bigr)$. We will define the functions $\psi'_1,\ldots,\psi'_n$ by defining the corresponding sequence of $M$ values $\psi'_{i,j}$ via a set of linear constraints as follows: For every $j_1,\ldots,j_n\in\bigl\{0,\ldots,\lfloor\nicefrac{H}{\varepsilon}\rfloor\bigr\}$ and for every two distinct subsets $N_1,N_2\subseteq N$, we add the following constraint:
\begin{center}
\begin{tabular}{lcr}
$\sum_{i\in N_1}\psi'_{i,j_i} = \sum_{i\in N_2}\psi'_{i,j_i}$ & \qquad if & \qquad $\sum_{i\in N_1}\phi'_{i,j_i} = \sum_{i\in N_2}\phi'_{i,j_i}$, \\
$\sum_{i\in N_1}\psi'_{i,j_i} \le \sum_{i\in N_2}\psi'_{i,j_i} - 1$ & \qquad if & \qquad $\sum_{i\in N_1}\phi'_{i,j_i} < \sum_{i\in N_2}\phi'_{i,j_i}$, \\
$\sum_{i\in N_1}\psi'_{i,j_i} \ge \sum_{i\in N_2}\psi'_{i,j_i} + 1$ & \qquad if & \qquad $\sum_{i\in N_1}\phi'_{i,j_i} > \sum_{i\in N_2}\phi'_{i,j_i}$. \\
\end{tabular}
\end{center}
We note that the set of constraints that we have defined is satisfiable. Indeed, setting
\[d\eqdef\min\left\{\biggl|\sum_{i\in N_1}\phi'_{i,j_i} - \sum_{i\in N_2}\phi'_{i,j_i}\biggr| ~\middle|~\begin{matrix}j_1,\ldots,j_n\in\bigl\{0,\ldots,\lfloor\nicefrac{H}{\varepsilon}\rfloor\bigr\} \\ \And N_1,N_2\subseteq N\end{matrix} \And \sum_{i\in N_1}\phi'_{i,j_i} \ne \sum_{i\in N_2}\phi'_{i,j_i}\right\},\]
where by finiteness $d\!>\!0$, we note that the assignment $\psi'_{i,j}=\frac{\phi'_{i,j}}{d}$ for all $i,j$ satisfies this set of constraints. Since this set of (exponentially many) constraints in $M$ unknowns is satisfiable, there exists a solution defined by precisely $M$ of these constraints, where all $M$ constraints are binding (i.e., hold with an equality rather than with an inequality). Let~$(\psi'_{i,j})_{i,j}$ be this solution. By Cramer's Rule, each $\psi'_{i,j}$ is the quotient of the determinants of two $M\times M$ matrices, where in our case, by construction each of these matrices contains only values in $\{-1,0,1\}$, and therefore, each such determinant is an integer having absolute value at most $M!$. We conclude the proof as by construction, for every valuation profile, the order of the ironed virtual welfares of every two possible \mbox{deterministic outcomes} in $\{0,1\}^n$ is the same according to both $(\phi'_i)_{i\in N}$ and $(\psi'_i)_{i\in N}$, and therefore $\bigl((\psi'_i)_{i\in N};\MAX_X\bigr)=\bigl((\phi'_i)_{i\in N};\MAX_X\bigr)$, as required.
\end{proof}

Plugging in \cref{myerson-single-parameter,exists-single-parameter,efficient-single-parameter,small-single-parameter} into the analysis of the previous \lcnamecrefs{convergence}, we obtain the following analogue of \cref{empirical}, providing a polynomial-time algorithm for learning an approximately optimal auction for a class of single-parameter environments that includes deterministic environments (including matroid environments and the public project environment) and position environments, from samples from an arbitrary unknown bounded product distribution.

\begin{theorem}[Detailed version of \cref{intro-empirical-single-parameter}]\label{empirical-single-parameter}
Let $X\subseteq[0,1]^n$ be a tractable deterministic environment (such as a matroid environment or the public project environment) or a position environment.
There exist $t\le s=\poly(H,n,\nicefrac{1}{\varepsilon},\log\nicefrac{1}{\delta})$ such that the following holds.
Let $F_1,\ldots,F_n$ be arbitrary distributions on $[0,H]$.
For every $i\in N$, draw $s$~samples from $F_i$, and let $\hat{F}_i$ be the empirical distribution over the first $t$ of the $s$ samples from $F_i$.
Then, with probability at least~$1\!-\!\delta$, the auction output by the deterministic polynomial-time ``$\frac{\varepsilon}{W_X+3}$-rounding'' algorithm of \cref{efficient-single-parameter} given the $s$ samples from $F=F_1\times\cdots\times F_n$ and given the optimal auction for $\hat{F}=\hat{F}_1\times\cdots\times\hat{F}_n$ (which by \cref{myerson-single-parameter} can be deterministically computed in time $\poly(t)$ via \cref{elkind}),
approximates the maximum possible revenue from~$F$ up to less than an additive~$\varepsilon$.
\end{theorem}

\begin{remark}\label{empirical-single-parameter-lt}
\cref{empirical-single-parameter} also holds, via the same proof, for every tractable single-parameter environment $X\subseteq[0,1]^n$  such that $\bigl|S^X_{\varepsilon}\bigr|\le\exp\bigl(\poly(H,n,\nicefrac{1}{\varepsilon})\bigr)$.
\end{remark}

\subsection{Computationally Hard Auction Environments}\label{intractable}

Up until now, we have assumed that the given single-parameter environment $X\subseteq[0,1]^n$ is tractable. Nonetheless, for many single-parameter environments, including deterministic environments, this is known not to be the case. We therefore now extend our analysis to such environments.

A notable example of an intractable single-parameter environment is that of knapsack environments (\citealp{MualemNisan}; see also \citealp{RoughgardenBook}), where each bidder~$i$ has a nonnegative ``size'' $w_i$, and a deterministic outcome is possible if the cumulative size of all winners does not exceed some threshold~$w$:
\[X_{(w_1,\ldots,w_n;w)}=\left\{(x_1,\ldots,x_n)\in\{0,1\}^n~\middle|~\sum_{i=1}^n w_i\cdot x_i \le w\right\}.\]
In such environments, to precisely maximize ironed virtual welfare (and therefore revenue), one must solve an instance of the well known KNAPSACK problem, which is known to be NP-complete.

A popular approach in the literature for designing auctions for intractable environments $X$ is to have the designed auction choose an outcome that approximately maximizes the ironed virtual welfare using some efficient constant-factor approximation algorithm $\APPROX_X$ instead of the intractable precise-maximization algorithm $\MAX_X$;
the payments are once again the unique nonnegative payments that ensure (precise) truthfulness\footnote{For this auction to be truthful, the approximation algorithm must also maintain a certain well understood monotonicity property. For details, see, e.g., \cite{HartlineBook,RoughgardenBook}.} for all valuation profiles in $[0,H]^n$. We denote this auction, which by tools developed by \cite{Myerson} turns out to approximate the revenue from the optimal auction $\OPT(F;X)=\bigl((\phi_i)_{i\in N};\MAX_X\bigr)$ up to the same constant multiplicative factor of the approximation algorithm, by $\bigl((\phi_i)_{i\in N};\APPROX_X\bigr)$.

Assume, therefore, that $X\subseteq[0,1]^n$ is an intractable environment and that $\APPROX_X$ is a deterministic algorithm that runs in time $\poly(H,n)$, such that given the (realized) ironed virtual values $\phi_i(v_i)$ of all bidders, $\APPROX_X$ outputs an outcome $(x_1,\ldots,x_n)\in X$ that maximizes the ironed virtual welfare ${\sum_{i=1}^n x_i\cdot \phi_i(v_i)}$ up to some multiplicative factor $C>1$. We will explore when our analysis can be applied to yield a polynomial-time algorithm for learning, from samples from an arbitrary unknown bounded product distribution, a tractable auction that approximates the optimal revenue up to the same multiplicative factor of $C$, plus less than an additive $\varepsilon$.

We start by noting that the analysis and results of \cref{existence-single-parameter} hold for all auctions, and in particular also auctions of the form $A=\bigl((\phi_i)_{i\in N};\APPROX_X\bigr)$.
Furthermore, we note that the analysis of \cref{efficiency-single-parameter} still holds when when replacing every occurrence of $\MAX_X$ with $\APPROX_X$ and every occurrence of $\Rev^{OPT(F;X)}(F)$ with $\frac{\Rev^{OPT(F;X)}(F)}{C}$.

\begin{proposition}\label{efficient-approx-single-parameter}
\cref{efficient-single-parameter} still holds when every occurrence of $\MAX_X$ is replaced with $\APPROX_X$, when $\Rev^{OPT(F;X)}(F)$ is replaced with $\frac{\Rev^{OPT(F;X)}(F)}{C}$, and when requiring only that $\APPROX_X$ runs in time $\poly(H,n)$ rather than requiring that $X$ is tractable.
\end{proposition}

Let $\OPT(\hat{F};X)=\bigl((\phi_i)_{i\in N};\MAX_X\bigr)$ be the optimal auction for the empirical distribution $\hat{F}=\hat{F}_1\times\cdots\times\hat{F}_n$ in the environment $X$. By assumption, the truthful auction $A=\bigl((\phi_i)_{i\in N};\APPROX_X\bigr)$ approximates the optimal revenue from $\hat{F}$ up to a multiplicative factor of $C$. By \cref{efficient-approx-single-parameter}, there exists an efficient algorithm that given $(\phi_i)_{i\in N}$ and $\hat{F}$ (and given the ability to draw some additional samples from $F$, if the algorithm is to be deterministic rather than randomized), outputs an $\varepsilon$-rounding $A'=\bigl((\phi'_i)_{i\in N};\APPROX_X\bigr)$ of $A$ such that
\[\Rev^{A'}(\hat{F})>\Rev^{A}(\hat{F})-(W_X+1)\cdot\varepsilon\ge\frac{\Rev^{\OPT(\hat{F};X)}(\hat{F})}{C}-(W_X+1)\cdot\varepsilon.\]
Now, if we had a result analogous to \cref{uniform}, showing that with high probability, for each $\varepsilon$-coarse auctions that uses $\APPROX_X$ to maximize some ironed virtual welfare as well as for $\OPT(F;X)$, its revenues on $F$ and on $\hat{F}$ differ by less than an additive $\varepsilon$, then we would obtain the desired approximation result via a derivation similar to that of \cref{empirical-single-parameter} (and \cref{empirical}):
\begin{align*}
\Rev^{\OPT(F;X)}(F)&<\Rev^{\OPT(F;X)}(\hat{F}) + \varepsilon \le \\
&\le\Rev^{\OPT(\hat{F};X)}(\hat{F})+\varepsilon<\\
&< C\cdot\bigl(\Rev^{A'}(\hat{F})+(W_X+1)\cdot\varepsilon\bigr)+\varepsilon< \\
&< C\cdot\bigl(\Rev^{A'}(F)+(W_X+2)\cdot\varepsilon\bigr)+\varepsilon= \\
&= C\cdot\bigl(\Rev^{A'}(F)+(W_X+2+\nicefrac{1}{C})\cdot\varepsilon\bigr)\le \\
&\le C\cdot\bigl(\Rev^{A'}(F)+(W_X+3)\cdot\varepsilon\bigr).
\end{align*}
While a first glance may suggest that \cref{small-single-parameter} indeed gives the desired analogue of \cref{uniform}, a closer look shows that more care is required here. Indeed, the fact that, as in the proof of \cref{small-single-parameter}, $\bigl((\psi'_i)_{i\in N};\MAX_X\bigr)=\bigl((\phi'_i)_{i\in N};\MAX_X\bigr)$ by no means implies that $\bigl((\psi'_i)_{i\in N};\APPROX_X\bigr)=\bigl((\phi'_i)_{i\in N};\APPROX_X\bigr)$, as many approximation algorithms use more information from $(\phi'_i)_{i\in N}$ beyond merely the order of all possible ironed virtual welfares for all possible valuation profiles and outcomes. For instance, the $2$-approximation algorithm for the knapsack environment \citep{MualemNisan} also considers the order of the quotients $\bigl(\frac{\phi'_i(v_i)}{w_i}\bigr)_{i\in N}$, which cannot be inferred from the order of all possible ironed virtual welfares alone. In the specific case of this 2-approximation algorithm, this is not the end of the road since the set of all possible orders of these quotients is of size at most $\exp\bigl(\poly(H,n,\nicefrac{1}{\varepsilon})\bigr)$.
More generally, we can obtain our approximation result whenever, for every $\varepsilon$-coarse $(\phi'_i)_{i\in N}$, we can encode all of the information that $\APPROX_X$ uses from $(\phi'_i)_{i\in N}$, such that the set of all possible encodings (i.e., the set of distinct $\varepsilon$-coarse auctions that use $\APPROX_X$ to maximize some ironed virtual welfare) is of size at most $\exp\bigl(\poly(H,n,\nicefrac{1}{\varepsilon})\bigr)$. (This is the case, for example, if $\APPROX_X$ only reads the virtual valuations $\phi'_i(v_i)$ up to a precision of some polynomial number of bits.)

\begin{definition}[$S^{\APPROX_X}_{\varepsilon}$]
For every $\varepsilon>0$, we denote the set of all $\varepsilon$-coarse auctions of the form $\bigl((\phi_i)_{i\in N};\APPROX_X\bigr)$ for a given approximation algorithm $\APPROX_X$ for a given single-parameter environment $X\subseteq[0,1]^n$ by $S^{\APPROX_X}_{\varepsilon}$.
\end{definition}

\begin{theorem}[Generalization of \cref{empirical-single-parameter}/\cref{intro-empirical-single-parameter} for Intractable Auction Environments]\label{empirical-approx-single-parameter}
Let $X\subseteq[0,1]^n$ be a single-parameter environment and let $\APPROX_X$ be a monotone algorithm that runs in time $\poly(H,n)$ and finds an outcome in $X$ that maximizes ironed virtual welfare up to some multiplicative factor $C\ge1$. If $\bigl|S^{\APPROX_X}_{\varepsilon}\bigr|\le\exp\bigl(\poly(H,n,\nicefrac{1}{\varepsilon})\bigr)$, then there exist $t\le s=\poly(H,n,\nicefrac{1}{\varepsilon},\log\nicefrac{1}{\delta})$ such that the following holds.
Let $F_1,\ldots,F_n$ be arbitrary distributions on $[0,H]$.
For every $i\in N$, draw $s$~samples from $F_i$, and let $\hat{F}_i$ be the empirical distribution over the first $t$ of the $s$ samples from $F_i$.
Then, with probability at least~$1\!-\!\delta$, the auction output by the deterministic polynomial-time ``$\frac{\varepsilon}{W_X+3}$-rounding'' algorithm of \cref{efficient-approx-single-parameter} given the~$s$ samples from $F=F_1\times\cdots\times F_n$ and given
$\bigl((\phi_i)_{i\in N};\APPROX_X\bigr)$, for $(\phi_i)_{i\in N}$ that can be deterministically computed in time $\poly(t)$ via \cref{elkind} given $\hat{F}=\hat{F}_1\times\cdots\times\hat{F}_n$,
approximates the maximum possible revenue from~$F$ up to the same multiplicative factor of $C$, plus less than an additive~$\varepsilon$.
\end{theorem}
\noindent
We note that \cref{empirical-approx-single-parameter} is a strict generalization of \cref{empirical-single-parameter}. Indeed, fixing ${C=1}$ in \cref{empirical-approx-single-parameter} yields \cref{empirical-single-parameter}, as a polynomial-time $\APPROX_X$ that guarantees revenue maximization ``up to a multiplicative factor of $1$'' is precisely a tractable~$\MAX_X$.

\section{Acknowledgments}

\begin{sloppypar}
Yannai Gonczarowski is supported by the Adams Fellowship Program of the Israel Academy of Sciences and Humanities.
The work of Noam Nisan was supported by ISF grant 1435/14 administered by the Israeli Academy of Sciences and by
Israel-USA Bi-national Science Foundation (BSF) grant number 2014389.
We thank Jamie Morgenstern and Tim Roughgarden for explaining to us some of the subtleties of the problem of efficient empirical revenue maximization in single-item auctions.
\end{sloppypar}

\bibliographystyle{abbrvnat}
\bibliography{empirical}

\begin{thebibliography}{15}
\providecommand{\natexlab}[1]{#1}
\providecommand{\url}[1]{\texttt{#1}}
\expandafter\ifx\csname urlstyle\endcsname\relax
  \providecommand{\doi}[1]{doi: #1}\else
  \providecommand{\doi}{doi: \begingroup \urlstyle{rm}\Url}\fi

\bibitem[Babichenko et~al.(2017)Babichenko, Barman, and
  Peretz]{Babichenko-Barman-Peretz}
Y.~Babichenko, S.~Barman, and R.~Peretz.
\newblock Empirical distribution of equilibrium play and its testing
  application.
\newblock \emph{Mathematics of Operations Research}, 42\penalty0 (1):\penalty0
  15--29, 2017.
\newblock Preliminary version (``Simple approximate equilibria in large
  games'') in Proceedings of the 15th ACM Conference on Economics and
  Computation (EC), 2014.

\bibitem[Cole and Roughgarden(2014)]{Cole-Roughgarden}
R.~Cole and T.~Roughgarden.
\newblock The sample complexity of revenue maximization.
\newblock In \emph{Proceedings of the 46th Annual ACM Symposium on Theory of
  Computing (STOC)}, pages 243--252, 2014.

\bibitem[Devanur et~al.(2016)Devanur, Huang, and Psomas]{Devanur-Huang-Psomas}
N.~R. Devanur, Z.~Huang, and C.-A. Psomas.
\newblock The sample complexity of auctions with side information.
\newblock In \emph{Proceedings of the 48th Annual ACM Symposium on Theory of
  Computing (STOC)}, pages 426--439, 2016.

\bibitem[Dughmi et~al.(2014)Dughmi, Han, and Nisan]{Dughmi-Li-Nisan}
S.~Dughmi, L.~Han, and N.~Nisan.
\newblock Sampling and representation complexity of revenue maximization.
\newblock In \emph{Proceedings of the 10th Conference on Web and Internet
  Economics (WINE)}, pages 277--291, 2014.

\bibitem[Elkind(2007)]{Elkind}
E.~Elkind.
\newblock Designing and learning optimal finite support auctions.
\newblock In \emph{Proceedings of the 18th Annual ACM-SIAM Symposium on
  Discrete Algorithms (SODA)}, pages 736--745, 2007.

\bibitem[Hart and Nisan(2013)]{Hart-Nisan-b}
S.~Hart and N.~Nisan.
\newblock The menu-size complexity of auctions.
\newblock In \emph{Proceedings of the 14th ACM Conference on Electronic
  Commerce (EC)}, page 565, 2013.

\bibitem[Hartline(2016)]{HartlineBook}
J.~D. Hartline.
\newblock \emph{Mechanism Design and Approximation}.
\newblock 2016.
\newblock Book draft.

\bibitem[Huang et~al.(2015)Huang, Mansour, and
  Roughgarden]{Huang-Mansour-Roughgarden}
Z.~Huang, Y.~Mansour, and T.~Roughgarden.
\newblock Making the most of your samples.
\newblock In \emph{Proceedings of the 16th ACM Conference on Economics and
  Computation (EC)}, pages 45--60, 2015.

\bibitem[Morgenstern(2016)]{jamie-tutorial}
J.~Morgenstern.
\newblock Sample complexity of auction design.
\newblock Presentation slides, Algorithmic Game Theory and Data Science
  Tutorial, The 17th ACM Conference on Economics and Computation (EC), July
  2016.
\newblock URL
  \url{https://www.cis.upenn.edu/~jamiemor/papers/agt-ds-tutorial.pdf}.

\bibitem[Morgenstern and Roughgarden(2015)]{Morgenstern-Roughgarden}
J.~Morgenstern and T.~Roughgarden.
\newblock On the pseudo-dimension of nearly optimal auctions.
\newblock In \emph{Proceedings of the 29th Annual Conference on Neural
  Information Processing Systems (NIPS)}, pages 136--144, 2015.

\bibitem[Mu'alem and Nisan(2008)]{MualemNisan}
A.~Mu'alem and N.~Nisan.
\newblock Truthful approximation mechanisms for restricted combinatorial
  auctions.
\newblock \emph{Games and Economic Behavior}, 64\penalty0 (2):\penalty0
  612--631, 2008.

\bibitem[Myerson(1981)]{Myerson}
R.~Myerson.
\newblock Optimal auction design.
\newblock \emph{Mathematics of Operations Research}, 6\penalty0 (1):\penalty0
  58--73, 1981.

\bibitem[Roughgarden(2016)]{RoughgardenBook}
T.~Roughgarden.
\newblock \emph{Twenty Lectures on Algorithmic Game Theory}.
\newblock Cambridge University Press, 2016.

\bibitem[Roughgarden and Schrijvers(2016)]{ironing-in-the-dark}
T.~Roughgarden and O.~Schrijvers.
\newblock Ironing in the dark.
\newblock In \emph{Proceedings of the 17th ACM Conference on Economics and
  Computation (EC)}, pages 1--18, 2016.

\bibitem[von Neumann(1951)]{vN51}
J.~von Neumann.
\newblock Various techniques used in connection with random digits.
\newblock \emph{National Bureau of Standards Applied Mathematics Series},
  12:\penalty0 36--38, 1951.

\end{thebibliography}

\appendix

\section{Examples Omitted from the Main Text}\label{examples}

\begin{example}\label{triangle}
Fix $H\eqdef2$. For every $\varepsilon>0$, there exists a product of regular distributions $F=F_1\times F_2\in\Delta\bigl([0,H]\bigr)^2$ such that for every $\varepsilon$-coarse $2$-bidder auction $A'$, there exists a valuation profile $(v_1,v_2)\in [0,H]^2$ such that $r^{A'}(v_1,v_2)<r^{\OPT}(v_1,v_2)-\nicefrac{1}{13}$, where $\OPT=\OPT(F)$ is the Myersonian auction that maximizes the revenue from $F$.

Indeed, let $\varepsilon>0$ and assume without loss of generality that $\varepsilon<\nicefrac{1}{3}$ (otherwise, divide~$\varepsilon$ by some large-enough integer).
Let $\epsfloor{1}$ denote the largest integer multiple of~$\varepsilon$ that is not greater than $1$.
Let $F_1=U\bigl([0,2]\bigr)$ and let $F_2\in\Delta\bigl(\bigl[\epsfloor{1}-\varepsilon,\epsfloor{1}\bigr]\bigr)$ be the continuous distribution with density function $f_2(x)=\nicefrac{2}{\varepsilon^2}\cdot\bigl(x-\epsfloor{1}+\varepsilon\bigr)$. By \cref{myerson}, $\OPT=(\phi_1,\phi_2)$ for the (ironed) virtual valuations $\phi_1,\phi_2$ corresponding to $F_1,F_2$, respectively.
It is well known for the uniform distribution $F_1$ that the corresponding (ironed) virtual valuation satisfies $\phi_1(v_1)=2\cdot v_1-2$ for every $v_1\in[0,2]$.
The crux of this example is that $F_2$ is defined so that despite the fact that it is tightly concentrated, the image of the corresponding (ironed) virtual valuation~$\phi_2$ is $\bigl(-\infty,\epsfloor{1}\bigr]$. In particular, since $\varepsilon<\nicefrac{1}{3}$, there exist $v_2,w_2\in\bigl[\epsfloor{1}-\varepsilon,\epsfloor{1}\bigr)$ such that $\phi_2(v_2)=\nicefrac{1}{3}$ and $\phi_2(w_2)=\nicefrac{2}{3}$. Note that $\phi_1(v_1)=\nicefrac{1}{2}$ for $v_1=\nicefrac{5}{4}$, and $\phi_1(w_1)=1$ for $w_1=\nicefrac{3}{2}$. Therefore $r^{\OPT}(v_1,v_2)=\nicefrac{7}{6}$ (bidder~$1$ wins and pays her minimal winning bid of $\phi_1^{-1}(\nicefrac{1}{3})=\nicefrac{7}{6}$) and $r^{\OPT}(w_1,w_2)=\nicefrac{4}{3}$ (bidder~$1$ wins and pays her minimal winning bid $\phi_1^{-1}(\nicefrac{2}{3})=\nicefrac{4}{3}$).

Let $A'$ be an $\varepsilon$-coarse $2$-bidder auction. We will show that the proposition is satisfied for either $(v_1,v_2)$ or $(w_1,w_2)$. Assume that $r^{A'}(v_1,v_2)\ge r^{\OPT}(v_1,v_2)-\nicefrac{1}{13}$. Therefore, we have that $r^{A'}(v_1,v_2)>1$.
Since $v_2,w_2\in\bigl[\epsfloor{1}-\varepsilon,\epsfloor{1}\bigr)$, we therefore have that $r^{A'}(v_1,w_2)=r^{A'}(v_1,v_2)>1$ as well. Therefore, for the valuation profile $(v_1,w_2)$, the winner is bidder $1$ (since bidder $2$ never pays more than her value). So, when bidder $2$ bids $w_2$, the minimal winning bid of bidder $1$ is at most $v_1=\nicefrac{5}{4}$, and so bidder $1$ does not pay more than $\nicefrac{5}{4}$ if she wins against $w_2$. Therefore,
\[r^{A'}(w_1,w_2)\le \nicefrac{5}{4}=\nicefrac{4}{3}-\nicefrac{1}{12}=r^{\OPT}(w_1,w_2)-\nicefrac{1}{12}<r^{\OPT}(w_1,w_2)-\nicefrac{1}{13},\]
as claimed.
\end{example}

\begin{remark}\label{triangle-round-down}
The construction underlying \cref{triangle} also demonstrates that for small enough $\varepsilon>0$, there exists a product of regular distribution $F=F_1\times F_2\in\Delta\bigl([0,H]\bigr)^2$ such that if $A'$ is the auction obtained by simply rounding-down each bid to the nearest integer multiple of~$\varepsilon$ and applying the allocation rule of $\OPT=\OPT(F)$ to the rounded-down bids (while adapting the payments to ensure truthfulness),\footnote{I.e., having each winning bidder pay her minimal winning bid.} then ${\Rev^{A'}(F)<\Rev^{\OPT}(F)-\nicefrac{1}{8}}$.

Indeed, let $\varepsilon>0$ and assume that $\varepsilon<\nicefrac{1}{4}$. For $F=F_1\times F_2$ as in \cref{triangle}, the ``rounded-down'' auction $A'$ is such that bidder~$1$ wins if and only if $v_1\ge\epsceil{1}$, where $\epsceil{1}$ is the smallest integer multiple of~$\varepsilon$ that is not less than $1$, and bidder $2$ almost surely\footnote{I.e., with probability $1$.} loses. Therefore, $\Rev^{A'}(F)=\frac{2-\epsceil{1}}{2}\cdot\epsceil{1}<\nicefrac{1}{2}\cdot(1+\varepsilon)$.
We lower-bound the revenue from $\OPT$ by considering an auction that sells to bidder~$1$ if $v_1\ge\nicefrac{3}{2}$, and otherwise to bidder~$2$.
The revenue from this auction is $\nicefrac{1}{4}\cdot\nicefrac{3}{2}+\nicefrac{3}{4}\cdot\bigl(\epsfloor{1}-\varepsilon\bigr)>\nicefrac{9}{8}-\nicefrac{3}{2}\cdot\varepsilon$.
Therefore,
\[
\Rev^{\OPT}(F)-\Rev^{A'}(F)>\nicefrac{9}{8}-\nicefrac{3}{2}\cdot\varepsilon-\nicefrac{1}{2}\cdot(1+\varepsilon)=\nicefrac{5}{8}-2\cdot\varepsilon>\nicefrac{1}{8},
\]
as claimed.
\end{remark}

\begin{example}\label{tight}
For every $H\ge\varepsilon>0$ and $\varepsilon>\eta>0$, there exists $F\in\Delta\bigl([0,H)\bigr)$ such that for every $\varepsilon$-coarse $1$-bidder auction $A'$, it is the case that $\Rev^{A'}(F)\le\Rev^{\OPT(F)}(F)-\eta$.

Indeed, let $v\eqdef\frac{\varepsilon+\eta}{2}>\eta>0$ and let $p\eqdef\frac{\eta}{v}\in(0,1)$. Let $F\in\Delta\bigl(\{0,v\}\bigr)$ that assigns probability $p$ to the value~$v$ and probability $1-p$ to the value $0$. Three truthful auctions (allocation rules) exist for~$F$:
\begin{enumerate}
\item
the bidder always wins (and pays her minimal winning bid of $0$),
\item
the bidder wins only if she bids $v$ (and pays her minimal winning bid of $v$),\footnote{In fact, any payment in $[0,v]$ can be charged in this case, but for our analysis we will be interested in the auction of this form that achieves maximum revenue, and this auction charges $v$ if the bidder wins.} and
\item
the bidder never wins.
\end{enumerate}
The revenue (from $F$) of the first and third auctions is $0$, while the revenue of the second auction is $p\cdot v = \eta>0$. Therefore, $\OPT(F)$ is the second of these three auctions. Note, however, that since $\{0,v\}\in[0,\varepsilon)$, only the first and third auctions are $\varepsilon$-coarse, so the revenue of any $\varepsilon$-coarse auction from $F$ is $0$, i.e., the loss in revenue of any $\varepsilon$-coarse auction, compared to $\OPT(F)$, is an additive $\eta$ (and in particular, at least an additive~$\eta$), as claimed.
\end{example}

\begin{example}\label{sharpp-complete}
Let $F_1=F_2=\cdots=F_n=U\bigl(\{0,1\}\bigr)$. For every $w_2,w_3,\ldots,w_n\in\NN$ and $w\in\NN$, we define the $n$-bidder (DSIC and ex-post IR) auction $A'(w_2,w_3,\ldots,w_n;w)$ to be the auction in which bidder~$1$ wins and pays $1$ if and only if both $v_1=1$ and $\sum_{i=2}^n w_i\cdot v_i=w$ (otherwise, no bidder wins). Computing the revenue $\Rev^{A'(w_2,\ldots,w_n;w)}(F_1\times\cdots\times F_n)$ is \mbox{\#P-complete}.

We first note that $A'(w_2,\ldots,w_n;w)$ is indeed DSIC and ex-post IR. Indeed, all bidders except bidder~$1$ have utility $0$ regardless of their reported bids, and so have no incentive to misreport their values or to not participate. Furthermore, bidder~$1$ has utility $0$ when bidding $0$ and when truthfully bidding $1$, and utility $\le0$ when untruthfully bidding $1$.

Now, by definition, the revenue $\Rev^{A'(w_2,\ldots,w_n;w)}(F_1\times\cdots\times F_n)$ equals half of the probability that $\sum_{i=2}^n w_i\cdot v_i=w$ for $v_2,v_3,\ldots,v_n$ that are drawn independently and uniformly from $\{0,1\}$. This probability, in turn, equals exactly $2^{-(n-1)}$ times the number of subsets of $\{w_2,\ldots,w_n\}$ whose sum equals $w$. To compute the latter, one must solve an instance of the counting analogue of the well known SUBSET-SUM problem, which is known to be \#P-complete.
\end{example}

\section{Proofs Omitted from the Main Text}\label{proofs}

\subsection{Proofs for Section~\refintitle{existence}}

\begin{proof}[Proof of \cref{one-by-one}]
Assume without loss of generality that $i=1$.
We prove each \lcnamecref{one-by-one-others} separately.

\vspace{.5em}\noindent\emph{Proof of \crefsubpartonly{one-by-one}{others}{always}:}
When $v_1\notin \bigl[j\cdot\varepsilon,(j+1)\cdot\varepsilon\bigr)$, the rounding of bids of bidder~$1$ in the $\varepsilon$-interval with index $j$ affects neither the determination of the winner, nor the minimal winning bid (which equals the payment) of bidders other than bidder $1$ (e.g., bidder~$i'$), as required.

\vspace{.5em}\noindent\emph{Proof of \crefsubpartonly{one-by-one}{others}{amortized}:}
\[\EE_{v_1\sim F_1|_j}\left[\EE_{A'}r_{i'}^{A'}(v_1,\ldots,v_n)\right]=
\EE_{v^1_j\sim F_1|_j}\left[r_{i'}^{A}(v^1_j,v_2,v_3,\ldots,v_n)\right]=
\EE_{v_1\sim F_1|_j}\left[r_{i'}^A(v_1,v_2,\ldots,v_n)\right].\]

\noindent\emph{Proof of \crefsubpartonly{one-by-one}{i}{notin}:}
Assume that $w_1 \notin \bigl(j\cdot\varepsilon,(j+1)\cdot\varepsilon\bigr)$. Therefore, the minimal winning bid of bidder $1$ in $A'$, regardless of the realization of $v^1_j\in \bigl[j\cdot\varepsilon,(j+1)\cdot\varepsilon\bigr)$, remains $w_1$, and so for every $v_1\in [0,H]$, we surely have that $r_1^{A'}(v_1,\ldots,v_n)=r_1^A(v_1,\ldots,v_n)$.

\vspace{.5em}\noindent\emph{Proof of \crefsubpartonly{one-by-one}{i}{in}(\labelcref{one-by-one-i-in-always}):}
We reason by cases.
For every $v_1<j\cdot\varepsilon$, we have, regardless of the realization of $v^1_j\in \bigl[j\cdot\varepsilon,(j+1)\cdot\varepsilon\bigr)$, that bidder $1$ loses against $v_{-1}$ in both $A'$ and $A$.
For every $v_1\ge(j+1)\cdot\varepsilon$, we have, regardless of the realization of $v^1_j$, that bidder $1$ wins against $v_{-1}$ in both $A'$ and $A$.

\vspace{.5em}\noindent\emph{Proof of \crefsubpartonly{one-by-one}{i}{in}(\labelcref{one-by-one-i-in-amortized}):}
The winning probability of a bid $v_1\sim F_1|_j$ in $A$ is $\PP(v_1\ge w_1)=\PP_{v_1\sim F_1|_j}(v_1\ge w_1)$, and the winning probability of such a bid in $A'$ is $\PP(v^1_j\ge w_1)=\PP_{v^1_j\sim F_1|_j}(v^1_j\ge w_1)$,
i.e., the same probability, as required.

\vspace{.5em}\noindent\emph{Proof of \crefsubpartonly{one-by-one}{i}{in}(\labelcref{one-by-one-i-in-payment}):}
Since $w_1\in\bigl(j\cdot\varepsilon,(j+1)\cdot\varepsilon\bigr)$, we have that the minimal winning bid~$w'_1$ of bidder~$1$ in $A'$ is surely either $j\cdot\varepsilon$ or $(j+1)\cdot\varepsilon$.
Therefore, surely $|w'_1-w_1|<\varepsilon$.
\end{proof}

\begin{proof}[Proof of \cref{coarse-pij}]
By \crefpart{one-by-one}{others} the only revenue that is affected by the change from $A$ to $A'$ is the revenue from bidder $i$, and by \crefsubpart{one-by-one}{i}{notin}, this revenue is affected only when the minimal winning bid of bidder $i$ lies in $\bigl(j\cdot\varepsilon,(j+1)\cdot\varepsilon\bigr)$, and even then, by \crefsubpart{one-by-one}{i}{in}, this revenue is affected with the probability that bidder~$i$ wins, and even then it is changed by less than an additive $\varepsilon$. So, the change in the overall revenue is less than an additive~$\varepsilon\cdot p^i_j$, as required. 
\end{proof}

\begin{proof}[Proof of \cref{coarse-amortized}]
Assume without loss of generality that $i=1$.
Let $A''$ be the result of applying to $A$ (in arbitrary order) all $F_1$-randomized $\varepsilon$-rounding actions on all of bidder~$1$'s $\varepsilon$-intervals.
We start by noting that by \crefpart{one-by-one}{others},
\[\EE_{v_1\sim F_1|_j}\left[\EE_{A'} r_1^{A'}(v_1)\right]
=
\EE_{v_1\sim F_1|_j}\left[\EE_{A''} r_1^{A''}(v_1)\right].
\]
It therefore suffices to show that for every fixed $v_2,\ldots,v_n\in[0,H]$, it is the case that 
\[\left|\EE_{v_1\sim F_1|_j}\left[\EE_{A''} r_1^{A''}(v_1,\ldots,v_n)\right] - \EE_{v_1\sim F_1|_j} r_1^A(v_1,\ldots,v_n)\right|<\varepsilon\cdot\smashoperator[l]{\PP_{v_1\sim F_1|_j}}[\text{$1$ wins in $A$ against $v_{-1}$}].\]
Let $w_1$ be the minimal winning bid of bidder $1$ in $A$ when the other bidders bid $v_2,\ldots,v_n$. Let~$k$ be such that $w_1\in \bigl[k\cdot\varepsilon,(k+1)\cdot\varepsilon\bigr)$,
and let $A'''$ be the result of applying the \mbox{$F_1$-randomized} $\varepsilon$-rounding action on bidder $1$'s $\varepsilon$-interval $k$ to $A$. By
thinking of $A''$ as having been constructed from $A$ by first applying the randomized rounding action for~$k$ to yield $A'''$ (following which, the minimal winning bid of bidder $1$ surely becomes either $k\cdot\varepsilon$ or $(k+1)\cdot\varepsilon$) and then applying all randomized rounding actions for intervals other than $k$, we obtain by
 \crefsubpart{one-by-one}{i}{notin} that
\[
\EE_{v_1\sim F_1|_j} \left[\EE_{A''} r_1^{A''}(v_1,\ldots,v_n)\right]=
\EE_{v_1\sim F_1|_j} \left[\EE_{A'''} r_1^{A'''}(v_1,\ldots,v_n)\right].
\]
We conclude the proof as by \crefsubpart{one-by-one}{i}{in} we obtain that
\[
\left|\EE_{v_1\sim F_1|_j} \left[\EE_{A'''} r_1^{A'''}(v_1,\ldots,v_n)\right]-\EE_{v_1\sim F_1|_j} r_1^A(v_1,\ldots,v_n)\right|<\varepsilon\cdot\smashoperator[l]{\PP_{v_1\sim F_1|_j}}[\text{$1$ wins in $A$ against $v_{-1}$}].
\]

(We note that when $w_1=k\cdot\varepsilon$, \crefsubpart{one-by-one}{i}{notin} implies that there is surely no difference between the two revenues above, and so \cref{coarse-amortized} can in fact be slightly strengthened, by replacing $p$ in the statement of this \lcnamecref{coarse-amortized} with the probability that $i$ wins in~$A$ \emph{and pays a price that is not an integer multiple of $\varepsilon$} when bidding $v_i\sim F_i|_j$, when the remaining bids are distributed according to $F_{-i}$.)
\end{proof}

\begin{proof}[Proof of \cref{randomized-rounding}]
By the triangle inequality and by \cref{coarse-amortized},
\[\left|\EE_{A'}\Rev^{A'}(F)-\Rev^A(F)\right|\le
\sum_{i=1}^n\left|\EE_{v_i\sim F_i} \left[\EE_{A'} r_i^{A'}(v_i)-r_i^{A}(v_i)\right]\right|
<
\sum_{i=1}^n\varepsilon\cdot\PP_F[\text{$i$ wins}]
=
p\cdot\varepsilon,\]
as required.

Alternatively, we may prove the result via \cref{coarse-pij}:
We start with $A$ and gradually apply, for all $i$ and $j$, the $F_i$-randomized $\varepsilon$-rounding action on bidder $i$'s interval $j$, to obtain $A'$.
By \cref{coarse-pij}, when the randomized rounding action for a pair $(i,j)$ is applied, the revenue changes by less than an additive $p^i_j\cdot\varepsilon$, and so less than an additive~$\varepsilon\cdot \PP_{F}\left[\text{$i$ wins}\right]$ for all randomized rounding actions for each bidder $i$, for a total of less than an additive $p\cdot\varepsilon$ for the entire randomized rounding rule. (Note that~$p^i_j$, in expectation over all randomized rounding actions already applied to bidders other than~$i$, does not change during the application of the randomized rounding action for any bidder-index pair other than $(i,j)$.)

(We note that regardless of whether we prove \cref{randomized-rounding} via \cref{coarse-pij} or via \cref{coarse-amortized},\footnote{For the latter, by the note concluding its proof.} \cref{randomized-rounding} can in fact be slightly strengthened, by replacing $p$ in the statement of this \lcnamecref{randomized-rounding} with the probability that some bidder wins in $A$ \emph{and pays a price that is not an integer multiple of $\varepsilon$} when the profile of bids is distributed according to~$F$.)
\end{proof}

\subsection{Proofs for Section~\refintitle{convergence}}

\begin{proof}[Proof of \cref{converge-one}]
We will show that with probability at least $1\!-\!\delta$, it is the case that $\bigl|\EE_{\hat{F}}r - \EE_{F}r\bigr|<\varepsilon$ for any fixed random variable $r:[0,H]^n\rightarrow[0,H]$. The \lcnamecref{converge-one} will follow by setting $r\eqdef r^A$.

Let $(v^1_1,\ldots,v^1_n),\ldots,(v^t_1,\ldots,v^t_n)$ be $t$ independently drawn random tuples sampled from $F$.
Let $k_1\eqdef1$ and for every $i\in\{2,3,\ldots,n\}$, fix $k_i\in\{1,\ldots,t\}$. We note that $(v^{k_1+j}_1,\ldots,v^{k_n+j}_n)$ for $j=1,\ldots,t$ (where addition of $k_i$ and $j$ wraps around from $t$ to $1$) are $t$ \emph{independent} random samples from $F_1\times\cdots\times F_n=F$.
Therefore, letting $r_{k_2,\ldots,k_n}\eqdef \frac{1}{t}\sum_{j=1}^t r(v^{k_1+j}_1,\ldots,v^{k_n+j}_n)$, we have by the Chernoff-Hoeffding Inequality that
\[\PP\Bigl(\bigl|r_{k_2,\ldots,k_n}-\EE_{F}r\bigr|\ge \varepsilon\Bigr) \le 2\exp\left(-\frac{2t\varepsilon^2}{H^2}\right).\]
Noting that by definition $\EE_{\hat{F}}r=\frac{1}{t^{n-1}}\sum_{k_2,\ldots,k_n\in\{1,\ldots,t\}}r_{k_2,\ldots,k_n}$, taking the union bound, we have that
\[\PP\Bigl(\bigl|\EE_{\hat{F}}r-\EE_{F}r\bigr|\ge\varepsilon\Bigr) \le t^{n-1}\cdot2\exp\left(-\frac{2t\varepsilon^2}{H^2}\right).\]
We conclude the proof as there exists $t=\poly(H,n,\nicefrac{1}{\varepsilon},\log\nicefrac{1}{\delta})$ such that \[t^{n-1}\cdot2\exp\left(-\frac{2t\varepsilon^2}{H^2}\right)\le\delta.\qedhere\]
\end{proof}

\begin{proof}[Proof of \cref{uniform}]
Let $\delta'\eqdef\frac{\delta}{|S^n_{\varepsilon}|+1}$. By \cref{small}, $\bigl|S^n_{\varepsilon}\bigr|\le\exp\bigl(\poly(H,n,\nicefrac{1}{\varepsilon})\bigr)$. Therefore, by \cref{converge-one} there exists $t=\poly(H,n,\nicefrac{1}{\varepsilon},\log\nicefrac{1}{\delta'})=\poly(H,n,\nicefrac{1}{\varepsilon},\log\nicefrac{1}{\delta})$ such that for every $A\in S^n_{\varepsilon}\cup\{\OPT(F)\}$ separately, with probability at least~$1-\delta'$, we have that
\[\bigl|\Rev^A(F) - \Rev^A(\hat{F})\bigr|<\varepsilon.\]
By the union bound, we therefore have that with probability at least~$1-\bigl(|S^n_{\varepsilon}|+1\bigr)\cdot\delta'=1-\delta$, this holds for \emph{all} auctions $A \in S^n_{\varepsilon}\cup\bigl\{\OPT(F)\bigr\}$ simultaneously,
as required.
\end{proof}

\subsection{Proofs for Section~\refintitle{single-parameter}}

\begin{proof}[Proof sketch of \cref{randomized-rounding-single-parameter}]
Similar to the proof of \cref{randomized-rounding}. Roughly speaking, instead of having a revenue change of less than $\varepsilon$ whenever bidder $i$ wins and her minimal winning bid is inside the open $\varepsilon$-interval to which the rounding action is applied (as in \cref{one-by-one}), it is the case that if $f$ is the fraction of the winnings of bidder $i$ whose threshold prices are inside this $\varepsilon$-interval, then the revenue change is less than $f'\cdot\varepsilon$, where~$f'$ is the ``part of $f$'' that $i$ wins.
\end{proof}

\begin{proof}[Proof of \cref{exists-single-parameter}]
Analogous to the proof of \cref{exists}, with \cref{randomized-rounding} replaced with \cref{randomized-rounding-single-parameter}.
\end{proof}

\begin{proof}[Proof of \cref{efficient-randomized-single-parameter}]
Let $A\eqdef\bigl((\phi_i)_{i\in N};\MAX_X\bigr)$. The algorithm is as follows, for $D,E=\poly(H,\nicefrac{1}{\varepsilon},\log\nicefrac{1}{\delta})$ that we define below:
\begin{enumerate}
\item
For every $e\in\{1,\ldots,E\}$:
\begin{itemize}
\item
Draw a valuation profile $(v^e_1,\ldots,v^e_n)\sim\hat{F}$.
\end{itemize}
\item
For every $d\in\{1,\ldots,D\}$:
\begin{enumerate}
\item
Draw a (deterministic) $\varepsilon$-rounding rule from (the distribution specified by) the $\hat{F}$-randomized $\varepsilon$-rounding rule.
\item
Apply this $\varepsilon$-rounding rule to $A$ to obtain an ($\varepsilon$-coarse) auction $A'_d$.
\item
Initialize $r_d\gets 0$.
\item
For every $e\in\{1,\ldots,E\}$:
\begin{itemize}
\item
Update\footnote{The explicit formula of \cite{Myerson} for the (truthful) payments in a Myersonian auction yields that given an $\varepsilon$-coarse Myersonian auction $A$ for a tractable single-parameter environment, and given a valuation profile $(v_1,\ldots,v_n)\in\supp\hat{F}$, it is possible in our setting to deterministically compute not only the allocation of $A$ in time $\poly(H,n)$, but also the associated payments in time $\poly(H,n,t)$.}
$r_d \gets r_d + \frac{r^{A'_d}(v^e_1,\ldots,v^e_n)}{E}$.
\end{itemize}
\end{enumerate}
\item
Output (the ironed virtual valuations of) the auction $A'_d$ for $d=\arg\Max_{d\in\{1,\ldots,D\}} r_d$.
\end{enumerate}

We now choose $D$ and $E$. We would like to choose $D$ such that with high probability, at least one of the drawn auctions $A'_d$ has revenue not significantly lower than the expected revenue of the $\hat{F}$-randomized $\varepsilon$-rounding of $A$, which we denote by $A'$. More accurately, we would like $D$ to satisfy the following:
\begin{equation}\label{efficient-single-parameter-D}
\PP\left(\tfrac{1}{D}\sum_{d=1}^D\Rev^{A'_d}(\hat{F}) \le \EE_{A'}\Rev^{A'}(\hat{F}) -\nicefrac{\varepsilon}{3}\right) \le \frac{\delta}{2}.
\end{equation}
By the Chernoff-Hoeffding Inequality, we have that
\[\PP\left(\tfrac{1}{D}\sum_{d=1}^D\Rev^{A'_d}(\hat{F}) \le \EE_{A'}\Rev^{A'}(\hat{F})-\nicefrac{\varepsilon}{3}\right) \le \exp\left(-\frac{2D\varepsilon^2}{9H^2}\right).\]
Therefore, by choosing $D\eqdef\frac{9H^2}{2\varepsilon^2}\log{\frac{2}{\delta}}=\poly(H,\nicefrac{1}{\varepsilon},\log\nicefrac{1}{\delta})$, we have that \cref{efficient-single-parameter-D} is satisfied.

We would now like to choose $E$ such that with high probability, the estimated revenue~$r_d$ for each $A'_d$ is close to the true revenue of $A'_d$ from $\hat{F}$. More accurately, we would like~$E$ to satisfy the following, for every $d\in\{1,\ldots,D\}$:
\begin{equation}\label{efficient-single-parameter-E}
\PP\Bigl(\bigl|r_d-\Rev^{A'_d}(\hat{F})\bigr| \ge \nicefrac{\varepsilon}{3}\Bigr) \le \frac{\delta}{2D}.
\end{equation}
By the Chernoff-Hoeffding Inequality, we have that
\[\PP\Bigl(\bigl|r_d-\Rev^{A'_d}(\hat{F})\bigr| \ge \nicefrac{\varepsilon}{3}\Bigr) \le 2\exp\left(-\frac{2E\varepsilon^2}{9H^2}\right).\]
Therefore, by choosing $E\eqdef\frac{9H^2}{2\varepsilon^2}\log\frac{4D}{\delta}=\poly(H,\nicefrac{1}{\varepsilon},\log\nicefrac{1}{\delta})$, we have that \cref{efficient-single-parameter-E} is satisfied for every $d\in\{1,\ldots,D\}$.

Having chosen $D$ and $E$, we now prove the correctness of the algorithm.
By taking the union bound over \cref{efficient-single-parameter-D} and over \cref{efficient-single-parameter-E} for all $d\in\{1,\ldots,D\}$, we obtain that with probability at least $1-\delta$,
both of the following hold:
\begin{itemize}
\item
For every $d\in\{1,\ldots,D\}$ it is the case that $\bigl|r_d-\Rev^{A'_d}(\hat{F})\bigr| < \nicefrac{\varepsilon}{3}$, and
\item
There exists $d'\in\{1,\ldots,D\}$ such that $\Rev^{A'_{d'}}(\hat{F})>\EE_{A'}\Rev^{A'}(\hat{F})-\nicefrac{\varepsilon}{3}$.
\end{itemize}
Let $d\in\{1,\ldots,D\}$ be such that $r_d$ is highest (so the algorithm outputs $A'_d$). In particular, $r_d\ge r_{d'}$. Therefore,
\begin{multline*}
\Rev^{A'_d}(\hat{F}) > r_d - \nicefrac{\varepsilon}{3} \ge r_{d'} - \nicefrac{\varepsilon}{3} > \Rev^{A'_{d'}}(\hat{F})-2\cdot\nicefrac{\varepsilon}{3} > \\
> \EE_{A'}\Rev^{A'}(\hat{F})-3\cdot\nicefrac{\varepsilon}{3}=\EE_{A'}\Rev^{A'}(\hat{F})-\varepsilon > \Rev^A(\hat{F})-(W_X+1)\cdot\varepsilon,
\end{multline*}
where the last inequality is by \cref{randomized-rounding-single-parameter}.

Finally, we estimate the number of random bits used by the algorithm. To draw $E$~valuation profiles, the algorithm requires $O(E\cdot n\cdot\log t)$ random bits. To draw $D$~roundings of $A'$, the algorithm requires $O\bigl(D\cdot n\cdot \nicefrac{H}{\varepsilon}\cdot\log t\bigr)$ random bits. Therefore, in total the algorithm uses at most $\poly(H,n,\nicefrac{1}{\varepsilon},\log\nicefrac{1}{\delta},\log t)$ random bits, as required.
\end{proof}

\begin{proof}[Proof of \cref{efficient-single-parameter}]
Let $b$ be as in \cref{efficient-randomized-single-parameter} for (the given) $t$ and~$\varepsilon$, and for guaranteeing success with probability at least $1-\nicefrac{\delta}{2}$.
We start by finding a probability~$\eta$ such that both of the following hold:
\begin{itemize}
\item
On the one hand, if for every $i\in N$ the distribution $F_i$ has an atom $v_i\in[0,H]$ with probability $F_i(v_i)\ge1-\eta$, then the auction that extracts the entire social welfare of the outcome that maximizes the social welfare for $(v_1,\ldots,v_n)$, maximizes the overall revenue from~$F$ up to an additive $\varepsilon$. Furthermore, in this case with probability at least $1\!-\!\delta$ each such atom $v_i$ can be recovered by taking the most common sample value out of $s_v=\poly(\log n,\log\nicefrac{1}{\delta})$ (chosen below) independent samples from $F_i$.
\item
On the other hand, if the above does not hold (i.e., if there exists $i\in N$ such that for every $v_i\in[0,H]$ it is the case that $F_i(v_i)<1-\eta$), then one can with probability at least $1\!-\!\nicefrac{\delta}{2}$ extract $b$ random bits from $s_b=\poly(H,n,\nicefrac{1}{\varepsilon},\log\nicefrac{1}{\delta},\log t)$ (chosen below) independent pairs of (independent) samples from $F$.
\end{itemize}

Assume first that for every $i\in N$ there exists $v_i\in[0,H]$ such that $F_i(v_i)\ge1-\eta$. Taking the union bound, we have that $F(v_1,\ldots,v_n)\ge1-n\eta$. For every $i\in N$, let
\begin{equation}\label{x-welfare-extractor}
\phi'_i(v'_i)\eqdef\begin{cases} v_i & v'_i \ge v_i, \\ -\infty & v'_i< v_i.\end{cases}
\end{equation}
We note that
\begin{equation}\label{efficient-single-parameter-x}
r^{((\phi'_i)_{i\in N};\MAX_X)}(v_1,\ldots,v_n)=\sum_{i=1}^n x_i\cdot v_i\ge r^{OPT(F;X)}(v_1,\ldots,v_n),
\end{equation}
where $(x_1,\ldots,x_n)\in X$ is the outcome that maximizes the social welfare $\sum_{i=1}^n x_i\cdot v_i$ for $(v_1,\ldots,v_n)$. (The inequality in \cref{efficient-single-parameter-x} holds as the revenue can never exceed the social welfare.) Therefore, by definition of $(v_1,\ldots,v_n)$ we have that
\[
\Rev^{((\phi'_i)_{i\in N};\MAX_X)}(F) \ge \Rev^{OPT(F;X)}(F)-n\eta H.
\]
Therefore, choosing any $\eta<\frac{\varepsilon}{nH}$, we obtain that if for every $i\in N$ there exists $v_i\in[0,H]$ such that $F_i(v_i)\ge1-\eta$, then $\Rev^{((\phi'_i)_{i\in N};\MAX_X)}(F)>\Rev^{OPT(F;X)}(F)-\varepsilon$, for $(\phi'_i)_{i\in N}$ as defined in \cref{x-welfare-extractor}. We choose $\eta\eqdef\frac{\varepsilon}{(n+3)H}<\frac{\varepsilon}{nH}$ as this also conveniently guarantees that $\eta\le\nicefrac{1}{4}$, which we will use later. (If $\varepsilon\ge H$, then any revenue is at most $\varepsilon$ and the \lcnamecref{efficient-single-parameter} is trivially satisfied.)

Still under the assumption that for every $i\in N$ there exists $v_i\in[0,H]$ such that $F_i(v_i)\ge1-\eta$, let $\freq_{v_i}$ be (a random variable that equals) the fraction of samples that equal $v_i$ out of $s_v$ independent samples from~$F_i$. Recall that $\eta\le\nicefrac{1}{4}$. By the Chernoff-Hoeffding Inequality, we have that
\[\PP\bigl(\freq_{v_i}\le\nicefrac{1}{2})\le\PP\Bigl(\freq_{v_i} \le 1-\eta-\nicefrac{1}{4}\Bigr) \le \exp\left(-\frac{s_v}{8}\right),\]
and so taking $s_v\eqdef\lceil8\log\nicefrac{n}{\delta}\rceil=\poly(\log n,\log \nicefrac{1}{\delta})$ guarantees that $\PP\bigl(\freq_{v_i}>\nicefrac{1}{2})\ge1-\nicefrac{\delta}{n}$,
and by the union bound we have that with probability at least $1-\delta$, for every $i\in N$ the most common sample value out of $s_v$ independent samples from $F_i$ is $v_i$.

\medskip

Assume now that there exists $i\in N$ such that for every $v_i\in[0,H]$ it is the case that $F_i(v_i)<1-\eta$. We now show that there exists $s_b=\poly(H,n,\nicefrac{1}{\varepsilon},\log\nicefrac{1}{\delta},t,b)$ such that with probability at least $1-\nicefrac{\delta}{2}$, from $s_b$ independent pairs of samples from $F_i$, at least $b$ pairs are of two nonidentical samples.
Since for every $v_i\in[0,H]$ it is the case that $F_i(v_i)<1-\eta$, there exists a partition of $[0,H]$ into two disjoint sets~$V,W$ such that $F_i(V)<1-\eta$ and $F_i(W)<1-\eta$.
As we will choose $s_b\eqdef s'\cdot b$, by the union bound we upper bound the probability that from $s$ independent pairs of samples from~$F_i$ less than $b$ pairs are of two nonidentical samples, by $b$ times the probability that $s'$ independent pairs of samples from~$F_i$ are all of identical samples. We upper bound the latter probability by the probability that $2s'$ independent samples from $F_i$ lie entirely in one of the sets $V,W$, which is less than $2(1-\eta)^{2s'}$. Choosing the smallest integer $s'$ such that $2(1-\eta)^{2s'}\le\frac{\delta}{2b}$ therefore guarantees that the probability that from $s_b$ independent pairs of samples from $F$ less than $b$ pairs are of two nonidentical samples, is at most $\nicefrac{\delta}{2}$. It therefore remains to estimate $s'$ (and hence $s_b$). By definition, $s'$ is the smallest integer such that $2s'\cdot\log(1-\eta)\le\log\frac{\delta}{4b}$, and so
\begin{multline*}
s' = \left\lceil\frac{\log\frac{\delta}{4b}}{2\log(1-\eta)}\right\rceil\le
\left\lceil\frac{\log\frac{\delta}{4b}}{-2\eta}\right\rceil=
\left\lceil\frac{\log\frac{4b}{\delta}}{2\eta}\right\rceil=
\left\lceil\frac{\log(4b)+\log\nicefrac{1}{\delta}}{2\eta}\right\rceil=\\
=\left\lceil\frac{(n+3)H\cdot\bigl(\log(4b)+\log\nicefrac{1}{\delta}\bigr)}{2\varepsilon}\right\rceil=
\poly(H,n,\nicefrac{1}{\varepsilon},\log\nicefrac{1}{\delta},\log b),
\end{multline*}
and so $s_b=s'\cdot b=\poly(H,n,\nicefrac{1}{\varepsilon},\log\nicefrac{1}{\delta},\log t)$.

\medskip

We are now ready to present the deterministic ``$\varepsilon$-rounding'' algorithm. Let $s\eqdef\Max\{2s_b,s_v,t\}=\poly(H,n,\nicefrac{1}{\varepsilon},\log\nicefrac{1}{\delta},t)$. The algorithm is as follows:
\begin{steps}
\item
Split the given $s$ samples from $F$ into $\lfloor\nicefrac{s}{2}\rfloor$ pairs of samples.
\item
If less than $b$ of the pairs are of two nonidentical samples:
\begin{steps}
\item
For each $i\in N$, let $v_i$ be the most common out of the $i$th-coordinate values of the $s$ samples.
\item\label{return-x-welfare-extractor}
Output $e=0$ and the functions $(\phi'_i)_{i\in N}$ where for every $i$,\ \ $\phi'_i$ is defined only for $v_i$, with $\phi'_i(v_i)\eqdef v_i$.\footnote{Recall that by \cref{extend-beyond-support}, as an ironed virtual valuation we interpret this $\phi'_i$ as coinciding with $\phi'_i$ as defined in \cref{x-welfare-extractor}.}
\end{steps}
\item
Else:
\begin{steps}
\item\label{compute-random-bits}
From each of the first $b$ of the pairs of samples that are of two nonidentical samples, compute a single (random) bit as follows: let $i$ be the first index such that the $i$th coordinates of these two samples differ; the desired (random) bit is $1$ if the $i$th coordinate of the first sample is smaller than the $i$th coordinate of the second sample, and $0$ otherwise.
\item\label{run-derandomized}
Deterministically run the algorithm from \cref{efficient-randomized-single-parameter} on $(\phi_i)_{i\in N}$ using the $b$ (random) bits computed above instead of the (at most) $b$ random bits required by the algorithm.
\item\label{return-derandomized}
Output $e=1$ and the functions $(\phi'_i)_{i\in N}$ that were output by the algorithm from \cref{efficient-randomized-single-parameter} when run in \cref{run-derandomized}.
\end{steps}
\end{steps}

We start by noting that if the above algorithm outputs $e=1$, then by \cref{efficient-randomized-single-parameter}, $\bigl((\phi'_i)_{i\in N};\MAX_X\bigr)$ is an $\varepsilon$-rounding of $\bigl((\phi_i)_{i\in N};\MAX_X\bigr)$.
To prove that the algorithm is correct with probability at least $1-\delta$, we first note that the algorithm can terminate in precisely one of two ways:
\begin{enumerate}
\item
In \cref{return-x-welfare-extractor}, since less than $b$ of the $\lfloor\nicefrac{s}{2}\rfloor$ pairs of samples are of two nonidentical samples.
\item
In \cref{return-derandomized}, since at least $b$ of the $\lfloor\nicefrac{s}{2}\rfloor$ pairs of samples are of two nonidentical samples.
\end{enumerate}
If the algorithm terminates in \cref{return-derandomized}, then since the $b$ bits computed in \cref{compute-random-bits} are indeed random and independent \citep{vN51}, and so \cref{efficient-randomized-single-parameter} guarantees that with probability at least $1-\nicefrac{\delta}{2}$, we have
\[\Rev^{((\phi'_i)_{i\in N};\MAX_X)}(\hat{F}) > \Rev^{((\phi_i)_{i\in N};\MAX_X)}(\hat{F})-(W_X+1)\cdot\varepsilon,\]
and therefore if the algorithm terminates in \cref{return-derandomized}, then it succeeds with at least this probability.
To complete the proof of correctness, we reason by cases, according to whether or not for every $i\in N$ there exists $v_i\in[0,H]$ such that $F_i(v_i)\ge1-\eta$.

We start by analyzing the case in which for every $i\in N$ there exists $v_i\in[0,H]$ such that $F_i(v_i)\ge1-\eta$. In this case, if the algorithm terminates in \cref{return-x-welfare-extractor}, then as shown above, with probability at least $1-\delta$ we correctly identify $v_i$ using $s\ge s_v$ samples, and so $e=0$ and (as shown above) $\Rev^{((\phi'_i)_{i\in N};\MAX_X)}(F)>\Rev^{OPT(F;X)}(F)-\varepsilon$, as required. So, regardless of whether the algorithm terminates in \cref{return-x-welfare-extractor} or in \cref{return-derandomized}, with probability at least $1-\delta$ the algorithm succeeds, as required.

It remains to analyze the case in which there exists $i\in N$ such that for every $v_i\in[0,H]$ it is the case that $F_i(v_i)<1-\eta$. In this case, as shown above, the probability that less than $b$ of the $\lfloor\nicefrac{s}{2}\rfloor\ge s_b$ pairs of samples have two nonidentical $i$th coordinate-values is at most $\nicefrac{\delta}{2}$. Therefore, with probability at least $1-\nicefrac{\delta}{2}$ the algorithm terminates in \cref{return-derandomized}, and as shown above, if this happens, then with probability at least $1-\nicefrac{\delta}{2}$, the algorithm succeeds. Therefore, taking the union bound, in this case with probability at least $1-\delta$ we have that $e=1$ and $\Rev^{((\phi'_i)_{i\in N};\MAX_X)}(\hat{F}) > \Rev^{((\phi_i)_{i\in N};\MAX_X)}(\hat{F})-(W_X+1)\cdot\varepsilon$, as required.
\end{proof}

\begin{proof}[Proof of \cref{empirical-single-parameter}]
The proof is analogous to that of \cref{empirical}.
An analogous version of \cref{uniform} can be proven for $S^X_{\varepsilon}$ rather than $S^n_{\varepsilon}$ via an analogous proof, replacing \cref{small} with the assumption that $\bigl|S^X_{\varepsilon}\bigr|\le\exp\bigl(\poly(H,n,\nicefrac{1}{\varepsilon})\bigr)$. (For position and matroid environments, no replacement is needed, and for arbitrary deterministic environments the assumption holds by \cref{small-single-parameter}.) Let $t$ be as in this analogue of \cref{uniform} for $S^X_{\varepsilon/(W_X+3)}$, for guaranteeing success with probability at least $1-\nicefrac{\delta}{2}$. Let~$s$ be as in \cref{efficient-single-parameter} for guaranteeing success with probability at least $1-\nicefrac{\delta}{2}$.
By \cref{elkind,efficient-single-parameter} (and by taking a union bound over two failure probabilities of at most $\nicefrac{\delta}{2}$ each), all that we have to show is that if the algorithm of \cref{efficient-single-parameter} outputs $e=1$ and an $\frac{\varepsilon}{W_X+3}$-rounded auction $A'\in S^n_{\varepsilon/(W_X+3)}$ that maximizes the revenue from $\hat{F}$ up to less than an additive $(W_X+1)\cdot\frac{\varepsilon}{W_X+3}$, then this auction also maximizes the revenue from~$F$ up to less than an additive $\varepsilon$. Showing this is analogous to the conclusion of the proof of \cref{empirical}.
\end{proof}

\begin{proof}[Proof of \cref{efficient-approx-single-parameter}]
Analogous to the proof of \cref{efficient-single-parameter,efficient-randomized-single-parameter}, replacing every occurrence of $\MAX_X$ with $\APPROX_X$,
replacing every occurrence of $\Rev^{OPT(F;X)}(F)$ with $\frac{\Rev^{OPT(F;X)}(F)}{C}$, and with \cref{efficient-single-parameter-x} becoming
\[
r^{((\phi'_i)_{i\in N};\APPROX_X)}(v_1,\ldots,v_n)=\sum_{i=1}^n x_i\cdot v_i\ge \frac{r^{OPT(F;X)}(v_1,\ldots,v_n)}{C},
\]
where $(x_1,\ldots,x_n)\in X$ is the outcome that $\APPROX_X$ outputs given $(v_1,\ldots,v_n)$, which maximizes the social welfare $\sum_{i=1}^n x_i\cdot v_i$ for $(v_1,\ldots,v_n)$ up to a multiplicative factor of~$C$.
\end{proof}

\begin{proof}[Proof of \cref{empirical-approx-single-parameter}]
Analogous (see \cref{intractable} for details) to the proof of \cref{empirical-single-parameter},
with $S^X_{\varepsilon/(W_X+3)}$ replaced with $S^{\APPROX_X}_{\varepsilon/(W_X+3)}$,
and with \cref{efficient-single-parameter} replaced with \cref{efficient-approx-single-parameter}.
\end{proof}

\section{Simple Efficient Empirical Revenue Maximization for I.I.D.\ Bidders}\label{iid}

Possibly the most na\"{i}ve way to construct a small set of auctions is to consider only auctions for which all of the parameters of the auction are specified with some fixed precision, and to learn an auction from this set for some empirical distribution by computing the optimal auction for that distribution, and then rounding its parameters to obtain an auction from this set. The best possible hope in this case would be that on each valuation profile for the bidders, the rounded auction loses less than an additive~$\varepsilon$ in revenue.
While, as shown in \cref{triangle} in \cref{examples}, such an approach turns out too good to be true in the general case of non-i.i.d.\ product distributions,
in this \lcnamecref{iid} we show that remarkably, this approach succeeds for arbitrary (even irregular) i.i.d.\ product distributions, reproving the existence of a polynomial-time learning algorithm from \cite{ironing-in-the-dark} via a considerably simpler argument and a more natural (albeit not dissimilar) algorithm,
and highlighting the difference between the general case considered in the main text and the i.i.d.\ case.

\subsection{Optimal Auctions}

We once again make very weak use of \citeauthor{Myerson}'s (\citeyear{Myerson}) characterization of optimal auctions for i.i.d.\ product distributions, and once again only present this characterization, which is a special case of that presented in \cref{definitions}, to the extent required by our analysis.

\begin{definition}[Second-Price Auction with Reserve Price and Ironed Intervals, \citealp{Myerson}]
A \emph{second-price auction with reserve price and ironed intervals} is a pair $(p,I)$, where $p\in[0,H]$ is called the \emph{reserve price}, and $I\subseteq\bigl\{[\ell,h)~\big|~p\le\ell<h\le H\bigr\}$ are (possibly infinitely many) pairwise-disjoint intervals called \emph{ironed intervals}.
In this auction, there is a winner unless all bidders bid below the reserve price, and the winner is the bidder with lowest index among those with the highest bid, unless the highest bid lies in the same ironed interval as lower bids, in which case the winner is the bidder with lowest index among those whose bid lies in this ironed interval; the winner pays her minimal winning bid.\footnote{Once again, this auction can also be made symmetric, by choosing the winner uniformly at random among all bidders with highest bid / whose bids lie in the same ironed interval as the highest bid, and adapting the payments accordingly. The revenue is the same either way. As noted above, in this paper we use lexicographic ordering for simplicity.} 
\end{definition}

\cite{Myerson} proved that for every continuous i.i.d.\ product distribution, there exists a second-price auction with reserve price and ironed intervals that obtains the optimal revenue. \cite{Elkind} showed the same for discrete i.i.d.\ product distributions, giving an efficient algorithm for computing this optimal auction. These results are special cases of \cref{myerson,elkind}, respectively.

\begin{theorem}[\citealp{Myerson}]\label{myerson-iid}
For every distribution $F\in\Delta\bigl([0,H]\bigr)$, there exists a second-price auction with reserve price and ironed intervals that for every $n\in N$ achieves maximum revenue from $F^n$ among all possible auctions. 
\end{theorem}

\begin{theorem}[\citealp{Elkind}]\label{elkind-iid}
Let $t\in\mathbb{N}$. There exists an algorithm that runs in time $\poly(t)$, such that given
a discrete distribution $\hat{F}\in\Delta\bigl([0,H]\bigr)$ with support of size at most~$t$, outputs a second-price auction with reserve price and ironed intervals that for every $n\in N$ achieves maximum revenue from $\hat{F}^n$ among all
possible auctions. 
\end{theorem}

The reader may verify that a Myersonian auction $(\phi_i)_{\in N}$ where $\phi_i=\phi$ for some function $\phi$ for all $i\in N$, coincides with a second-price auction with reserve price and ironed intervals, where the reserve price is $\phi^{-1}(0)$, and the ironed intervals are the intervals on which $\phi$ is constant and nonnegative.

\subsection{Natural Efficient Distribution-Independent Rounding of\texorpdfstring{\\}{ }Second-Price Auctions with Reserve Prices}\label{existence-and-efficiency-iid}

As already noted above, we will now show that rounding the parameters that define a second-price auction with reserve price and ironed intervals to some fixed precision results in a loss in revenue of at most that precision on each valuation profile.

\begin{definition}[$\varepsilon$-Rounded-Down Auction]
Let $\varepsilon>0$.
\begin{itemize}
\item
For every $x\in\RR$, let $\epsfloor{x}\eqdef\varepsilon\cdot\left\lfloor\nicefrac{x}{\varepsilon}\right\rfloor$ be the value of $x$, rounded down to the nearest integer multiple of $\varepsilon$.
\item
For every second-price auction with reserve price and ironing intervals $(p,I)$, we define its \emph{$\varepsilon$-rounded-down} counterpart $\epsfloor{(p,I)}$ as follows:
\[
\epsfloor{(p,I)} \eqdef \Bigl(\epsfloor{p},\bigl\{\bigl[\epsfloor{\ell},\epsfloor{h}\bigr) ~\big|~ [\ell,h)\in I \And \epsfloor{\ell}<\epsfloor{h}\bigr\}\Bigr).
\]
\end{itemize}
\end{definition}

\begin{lemma}\label{distribution-independent}
For every $(p,I)$, for every $\varepsilon>0$, for every $n\in\mathbb{N}$, and for every $v_1,\ldots,v_n\in [0,H]$, it is the case that \[r^{\epsfloor{(p,I)}}(v_1,\ldots,v_n)>r^{(p,I)}(v_1,\ldots,v_n)-\varepsilon.\]
\end{lemma}

The proof of \cref{distribution-independent} and of other results from this \lcnamecref{iid} are relegated to \cref{proofs-iid}.
\cref{distribution-independent} implies that the loss in revenue is less than an additive $\varepsilon$ also on every distribution.

\begin{proposition}\label{exists-and-efficient-iid}
For every $(p,I)$, for every $\varepsilon>0$, for every $F\in\Delta\bigl([0,H]\bigr)$ and for every $n\in\mathbb{N}$, it is the case that \[\Rev^{\epsfloor{(p,I)}}(F^n)>\Rev^{(p,I)}(F^n)-p\cdot\varepsilon,\]
where $p$ is the probability that some bidder wins in $(p,I)$ when the profile of bids is distributed according to $F^n$.\footnote{Interestingly, similarly (at least conceptually) to \cref{exists} (and \cref{exists-single-parameter}), a slightly stronger statement than that of \cref{exists-and-efficient-iid} also holds, where $p$ is replaced with the probability that some bidder wins in $A$ \emph{and pays a price that is not an integer multiple of $\varepsilon$} when the profile of bids is distributed according to~$F$. The case analysis is similar to the one in the proof of \cref{distribution-independent}, and is left to the reader.}
\end{proposition}

We note that \cref{distribution-independent,exists-and-efficient-iid} (and therefore also all of the results of \cref{convergence-iid}) would still hold (via similar analysis) even if we were to $\varepsilon$-round $\epsfloor{(p,I)}$ into an $\varepsilon$-coarse auction by rounding its inputs (in addition to the parameters of the auction) down to the nearest integer multiple of $\varepsilon$ before applying its allocation rule (while adapting the payments accordingly).
We further note that a construction similar to that of \cref{tight} in \cref{examples} shows that the bound of ``less than~$\varepsilon$'' on the revenue loss in \cref{exists-and-efficient-iid} (and hence also in \cref{distribution-independent}) cannot be unconditionally tightened any further, even if $\epsfloor{(p,I)}$ is replaced with an arbitrary $\varepsilon$-coarse auction, and already for $n\!=\!1$~bidder (more bidders with the same valuation distribution may be added to show the same result for $n>1$ bidders).

\subsection{Uniform Convergence over the Set of\texorpdfstring{\\}{ }Rounded Second-Price Auctions with Reserve Prices}\label{convergence-iid}

Analogously to the main text, we obtain a uniform convergence result through showing that the set of $\varepsilon$-rounded-down second-price auctions with reserve price and ironing intervals is small, and then utilize this and the previous results of this \lcnamecref{iid} to reprove the analogous learning result for i.i.d.\ produce distributions.

\begin{definition}[$\epsfloor{S}$]
For every $\varepsilon>0$,
we denote the set of all $\varepsilon$-rounded-down second-price auctions with reserve price and ironing intervals with parameters in $[0,H]$ by $\epsfloor{S}$.
\end{definition}

\begin{lemma}\label{small-iid}
$\bigl|\epsfloor{S}\bigr|\le\exp\bigl(\poly(H,\nicefrac{1}{\varepsilon})\bigr)$.
\end{lemma}

\begin{proof}
To specify the ironing intervals, it is more than sufficient to merely specify, for each multiple of $\varepsilon$ in $[p,H]$, whether an ironed interval ends at that point, starts at that point, neither, or both.
\end{proof}

To prove an analogous result to \cref{empirical}/\cref{intro-empirical} for i.i.d.\ distributions using the results of \cref{existence-and-efficiency-iid}, we require the following analogue of \cref{uniform}.

\begin{lemma}\label{uniform-iid}
For every $\varepsilon>0$ and $\delta>0$, there
exists $t=\poly(H,n,\nicefrac{1}{\varepsilon},\log\nicefrac{1}{\delta})$ such that the following holds.
Fix $F\in\Delta\bigl([0,H]\bigr)$,
draw $t$ samples from $F$, and let $\hat{F}$ be the empirical
uniform distribution over these $t$ samples. With probability at least~$1\!-\!\delta$, it is the case that
\[
\bigl|\Rev^A(F^n) - \Rev^A(\hat{F}^n)\bigr|<\varepsilon
\]
holds simultaneously for every $A \in \epsfloor{S}\cup\bigl\{\OPT(F^n)\bigr\}$.
\end{lemma}

The proof of \cref{uniform-iid} is analogous to that of \cref{uniform} (taking a union bound due to the exponential size of $\epsfloor{S}$), with \cref{converge-one} replaced with the following analogous \lcnamecref{converge-one-iid} for i.i.d.\ distributions. We phrase and prove the following \lcnamecref{converge-one-iid} as a general probabilistic concentration inequality, which we also find interesting in its own right. (When reading this \lcnamecref{converge-one-iid}, the reader is encouraged to think of the random variable $r$ as $r^A$ for some fixed auction $A$.)

\begin{proposition}\label{converge-one-iid}
For every $\varepsilon>0$ and $\delta>0$, there
exists $t=\poly(H,n,\nicefrac{1}{\varepsilon},\log\nicefrac{1}{\delta})$ such that the following holds.
Fix $F\in\Delta\bigl([0,H]\bigr)$ and fix $r:[0,H]^n\rightarrow[0,H]$.
Draw $t$~samples from $F$, and let $\hat{F}$ be the empirical
uniform distribution over these $t$ samples. Then, with probability
at least~$1\!-\!\delta$ it is the case that
\[
\bigl|\EE_{\hat{F}^n} r - \EE_{F^n} r\bigr|<\varepsilon.
\]
\end{proposition}

\begin{remark}
We note that \cref{converge-one-iid} does not follow from \cref{converge-one} (or from the more general concentration inequality of \citealp{Babichenko-Barman-Peretz} / \citealp{Devanur-Huang-Psomas}). Indeed, while in the latter a random tuple from $\hat{F}_1\times\cdots\times \hat{F}_n$ distributes according to $F_1\times\cdots\times F_n$, in the former a random tuple from $\hat{F}^n$ does not necessarily distribute according to $F^n$. Indeed, a random tuple from $\hat{F}^n$ has, e.g., positive probability for containing duplicate values, while a random tuple from $F^n$ may have (depending on $F$) zero probability for containing duplicate values.
\end{remark}

\begin{remark}
Concentration inequalities similar to \cref{converge-one,converge-one-iid} can be similarly shown to hold for cases in which the $n$ distributions (e.g., over valuations of bidders) are divided into subsets, where in each subset the distributions are i.i.d.
\end{remark}

Combining \cref{uniform-iid,exists-and-efficient-iid,elkind-iid}, we obtain the following analogue of \cref{empirical}/\cref{intro-empirical}, providing a natural polynomial-time algorithm for learning an approximately optimal auction from samples from an arbitrary unknown bounded i.i.d.\ product distribution.

\begin{theorem}\label{empirical-iid}
There exists $t=\poly(H,n,\nicefrac{1}{\varepsilon},\log\nicefrac{1}{\delta})$ such that the following holds.
Let $F$ be an arbitrary distribution on $[0,H]$.
Draw $t$ samples from $F$, and let $\hat{F}$ be the empirical distribution over these $t$ samples. Then, with probability at least~$1\!-\!\delta$, the optimal auction for $\hat{F}^n$ (which can be deterministically computed in time $\poly(t)$ via \cref{elkind-iid}), where the reserve price as well as the bounds of each ironed interval are all rounded down to the nearest multiple of $\nicefrac{\varepsilon}{3}$, approximates the maximum possible revenue from $F^n$ up to less than an additive~$\varepsilon$.
\end{theorem}

\subsection{Proofs for Appendix~\refintitle{iid}}\label{proofs-iid}

\begin{proof}[Proof \cref{distribution-independent}]
For a value $v\in[0,H]$, define
\[
h_{(p,I)}(v)\eqdef\begin{cases} p & v\le p \\ h & \exists [\ell,h)\in I: v\in[\ell,h) \\ v & \mbox{otherwise,}\end{cases}
\]
and
\[
\ell_{(p,I)}(v)\eqdef\begin{cases} p & v\le p \\ \ell & \exists [\ell,h)\in I: v\in[\ell,h) \\ v & \mbox{otherwise,}\end{cases}
\]
and note that for a valuation profile with second-highest bid $v^{(2)}$, the revenue from this profile, assuming that not all bids are lower than the reserve price $p$, is either $h_{(p,I)}(v^{(2)})$ or $\ell_{(p,I)}(v^{(2)})$.

Let $(v_1,\ldots,v_n)\in[0,H]^n$ be a valuation profile. To show that $r^{(p,I)}(v_1,\ldots,v_n)<r^{\epsfloor{(p,I)}}(v_1,\ldots,v_n)+\varepsilon$, we reason by cases. We note that if the revenue from $(p,I)$ is zero, then there is nothing to prove. We assume therefore that the highest bid is no less than~$p$, and so also no less than $\epsfloor{p}$.

\begin{itemize}
\item
If the winning bid in $(p,I)$ lies in the same ironed interval (in $(p,I)$) as another bid:

Let $v^{(2)}$ be the second-highest bid (which may or may not be the winning bid). In this case, $r^{\epsfloor{(p,I)}}(v_1,\ldots,v_n)\ge \ell_{\epsfloor{(p,I)}}(v^{(2)})>\ell_{(p,I)}(v^{(2)})-\varepsilon= r^{(p,I)}(v_1,\ldots,v_n)-\varepsilon$.

\item
If the winning bid in $(p,I)$ does not lie in the same ironed interval (in $(p,I)$) as another bid:

Let $v^{(2)}$ be the second-highest bid.
\begin{itemize}

\item
If the winning bid in $\epsfloor{(p,I)}$ lies in the same ironed interval (in $\epsfloor{(p,I)}$) as another bid:

This means that $v^{(2)}$ lies inside an ironed interval that ``before rounding'' did not contain it.
Therefore (regardless of whether or not the winning bidder changed), $r^{\epsfloor{(p,I)}}(v_1,\ldots,v_n)\ge \ell_{\epsfloor{(p,I)}}(v^{(2)})=\epsfloor{v^{(2)}} > h_{(p,I)}(v^{(2)})-\varepsilon\ge r^{(p,I)}(v_1,\ldots,v_n)-\varepsilon$, as required.

\item
If the winning bid in $\epsfloor{(p,I)}$ does not lie in the same ironed interval (in $\epsfloor{(p,I)}$) as another bid:

Recall that the revenue from $(p,I)$ is
$\ell_{(p,I)}(v^{(2)})$ if $v^{(2)}$ lies in an ironed interval (in $(p,I)$) containing bids only from bidders with indices higher than that of the winner, and otherwise $h_{(p,I)}(v^{(2)})$.
Similarly, the revenue from $\epsfloor{(p,I)}$ is
$\ell_{\epsfloor{(p,I)}}(v^{(2)})$ if $v^{(2)}$ lies in an ironed interval (in $\epsfloor{(p,I)}$) containing bids only from bidders with indices higher than that of the winner, and otherwise $h_{\epsfloor{(p,I)}}(v^{(2)})$.

Since $h_{\epsfloor{(p,I)}}(v^{(2)})>h_{(p,I)}(v^{(2)})-\varepsilon$ and $\ell_{\epsfloor{(p,I)}}(v^{(2)})>\ell_{(p,I)}(v^{(2)})-\varepsilon$ (and since $h_{(p,I)}(v^{(2)})\ge \ell_{(p,I)}(v^{(2)})$ and $h_{\epsfloor{(p,I)}}(v^{(2)})\ge \ell_{\epsfloor{(p,I)}}(v^{(2)})$), we need only verify that no significant revenue loss occurred if the winner pays $h_{(p,I)}(v^{(2)})>\ell_{(p,I)}(v^{(2)})$ in $(p,I)$ but pays $\ell_{\epsfloor{(p,I)}}(v^{(2)})<h_{\epsfloor{(p,I)}}(v^{(2)})$ in $\epsfloor{(p,I)}$. In this case, some bid $v'<v^{(2)}$ (of a bidder with low index) that lies in the same ironed interval in $(p,I)$ as $v^{(2)}$ does not lie in the same ironed interval as $v^{(2)}$ in $\epsfloor{(p,I)}$. Furthermore, in this case $v^{(2)}$ lies inside some ironed interval in $\epsfloor{(p,I)}$. This means that $v^{(2)}$ lies inside an ironed interval in $\epsfloor{(p,I)}$ that ``before rounding'' did not contain it. Therefore, $r^{\epsfloor{(p,I)}}(v_1,\ldots,v_n)=\ell_{\epsfloor{(p,I)}}(v^{(2)})=\epsfloor{v^{(2)}} >h_{(p,I)}(v^{(2)})-\varepsilon=r^{(p,I)}(v_1,\ldots,v_n)-\varepsilon$, as required.\qedhere
\end{itemize}
\end{itemize}
\end{proof}

\begin{proof}[Proof of \cref{exists-and-efficient-iid}]
Immediate from \cref{distribution-independent}, noting that no revenue loss occurs for valuation profiles for which the auction $(p,I)$ chooses no winner.
\end{proof}

\begin{proof}[Proof of \cref{converge-one-iid}]
The main idea behind the proof is similar to the one behind that of \cref{converge-one}, but keeping the samples independent is somewhat more delicate and involved.
Let $v^1,\ldots,v^t$ be $t$ independently drawn random values sampled from $F$.
Let \[K \eqdef \bigl\{(k_1,\ldots,k_n) ~\big|~ k_1=1 \And \forall i\in\{2,\ldots,n\}: k_i \in \{1,\ldots,t\}\setminus\{k_1,\ldots,k_{i-1}\}\bigr\},\]
and note that $|K|=\frac{(t-1)!}{(t-n)!}$.
Fix $(k_1,\ldots,k_n)\in K$.
Since $k_1,\ldots,k_n$ are distinct, we note that for every $j\in\{1,\ldots,t\}$, the tuple $(v^{k_1+j},\ldots,v^{k_n+j})$ (where addition of $k_i$ and $j$ wraps around from $t$ to $1$) is a random sample from $F^n$. We now greedily partition the set of $t$ $n$-tuples of indices $I=I(k_1,\ldots,k_n)\eqdef\{k_1+j,k_2+j,\ldots,k_n+j\}_{j=1}^t$ into sets such that in each such set, the indices in all coordinates of all $n$-tuples are distinct. Let $I_0$ be a maximum subset of $I$ where all indices are distinct, let $I_1$ be a maximum subset of $I\setminus I_0$ where all indices are distinct, let $I_2$ be a maximum subset of $I\setminus (I_0\cup I_1)$ where all indices are distinct, etc.

We first claim that $I_0\cup I_1\cup\cdots\cup I_{n\cdot(n-1)} = I$, i.e., that this process concludes after at most $n\!\cdot(n\!-\!1)\!+\!1$ greedy steps.
To show this, it is enough to show that $|I_i| \ge \frac{|I \setminus \{I_0,\ldots, I_{i-1}\}|}{n\cdot(n-1)+1-i}$ for every $i=0,\ldots,n\!\cdot\!(n\!-\!1)$, as this implies that $I_{n\cdot(n-1)}=I\setminus\{I_0,\ldots,I_{n\cdot(n-1)-1}\}$. Indeed, for each such $i$, each $n$-tuple in $I_i$ ``intersects'' at most $n\!\cdot\!(n\!-\!1)$ other $n$-tuples in~$I$ (for every $j$, the $j$th coordinate of this $n$-tuple ``intersects'' the $k$th coordinate of precisely one other $n$-tuple in $I$, for $k\ne j$), at least one of which is contained in each of $I_0,\ldots,I_{i-1}$ (otherwise, this $n$-tuple could have been added to one of these sets), and so each $n$-tuple in $I_i$ ``blocks'' at most $n\!\cdot\!(n\!-\!1)\!-\!i$ $n$-tuples from $I \setminus \{I_1,\ldots, I_{i-1}\}$ from being added to $I_i$, obtaining that $|I_i| \ge \frac{|I \setminus \{I_0,\ldots, I_{i-1}\}|}{n\cdot(n-1)+1-i}$ (under a worst-case scenario where the sets of ``blocked'' $n$-tuples are disjoint for every two ``blocking'' $n$-tuples), as required.

By construction, for every $i$, we have that $(v^{k_1+j},\ldots,v^{k_n+j})$ for all $j$ such that $({k_1+j},\ldots,{k_n+j})\in I_i$, are $|I_i|$ \emph{independent} random samples from $F^n$. Therefore, for every $i$ with $|I_i|\ge\sqrt{t}$, letting $r^i_{k_1,\ldots,k_n}\eqdef \frac{1}{|I_i|}\sum_{(k_1+j,\ldots,k_n+j)\in I_i} r(v^{k_1+j},\ldots,v^{k_n+j})$, we have by the Chernoff-Hoeffding Inequality that
\[\PP\Bigl(\bigl|r^i_{k_1,\ldots,k_n}-\EE_{F^n}r\bigr|\ge \nicefrac{\varepsilon}{2}\Bigr) \le 2\exp\left(-\frac{|I_i|\varepsilon^2}{2H^2}\right)\le 2\exp\left(-\frac{\sqrt{t}\varepsilon^2}{2H^2}\right).\]
Let $U(k_1,\ldots,k_n)\eqdef\bigcup_{i:|I_i|\ge\sqrt{t}} I_i$. By definition, for $t\ge\bigl(n\cdot(n-1)+1\bigr)\cdot\sqrt{t}$, there exists $i$ such that $|I_i|\ge\sqrt{t}$ and so we have that $\bigl|U(k_1,\ldots,k_n)\bigr|=|I|-\bigl|\bigcup_{i:|I|<\sqrt{t}}I_i\bigr|>t-n\cdot(n-1)\cdot\sqrt{t}$.

Let now $U\eqdef\bigcup_{(k_1,\ldots,k_n)\in K} U(k_1,\ldots,k_n)$.
Taking the union bound, we have that
\[\PP\Biggl(\biggl|\Bigl(\tfrac{1}{|U|}\cdot\smashoperator{\sum_{\substack{(k_1,\ldots,k_n)\in K \\ i:|I_i(k_1,\ldots,k_n)|\ge\sqrt{t}}}}\bigl|I_i|\cdot r^i_{k_1,\ldots,k_n}\Bigr)-\EE_{F^n}r\biggr|\ge\nicefrac{\varepsilon}{2}\Biggr) \le \frac{(t-1)!}{(t-n)!}\cdot\bigl(n\cdot(n-1)+1\bigr)\cdot2\exp\left(-\frac{\sqrt{t}\varepsilon^2}{2H^2}\right).\]
Noting that by definition,
\[
\EE_{\hat{F}^n}r=\frac{|U|}{t^n}\cdot\LaTeXoverbrace{\tfrac{1}{|U|}\cdot\smashoperator{\sum_{\substack{(k_1,\ldots,k_n)\in K \\ i:|I_i(k_1,\ldots,k_n)|\ge\sqrt{t}}}}|I_i|\cdot r^i_{k_1,\ldots,k_n}}^{\approx\EE_{F^n}r}
+
\left(1-\frac{|U|}{t^n}\right)\cdot\LaTeXoverbrace{\tfrac{1}{t^n-|U|}\cdot\smashoperator{\sum_{\substack{(k_1,\ldots,k_n) \in \\ \{1,\ldots,t\}^n\setminus U}}} r(v^{k_1},\ldots,v_n^{k})}^{\le H},
\]
and noting that $|U|>|K|\cdot \bigl(t-n\cdot(n-1)\cdot\sqrt{t}\bigr)=\frac{(t-1)!}{(t-n)!}\cdot \bigl(t-n\cdot(n-1)\cdot\sqrt{t}\bigr)$,
we conclude the proof as there exists $t=\poly(H,n,\nicefrac{1}{\varepsilon},\log\nicefrac{1}{\delta})$ such that both 
\[\frac{(t-1)!}{(t-n)!}\cdot\bigl(n\cdot(n-1)+1\bigr)\cdot2\exp\left(-\frac{\sqrt{t}\varepsilon^2}{2H^2}\right)\le\delta\qquad\text{and}\]
\[\frac{\frac{(t-1)!}{(t-n)!}\cdot\bigl(t-n\cdot(n-1)\cdot\sqrt{t}\bigr)}{t^n}\ge1-\frac{\varepsilon}{2H}.\qedhere\]
\end{proof}

\begin{proof}[Proof of \cref{uniform-iid}]
Analogous to the proof of \cref{uniform}, with $S^n_{\varepsilon}$ replaced with $\epsfloor{S}$, and with \cref{small,converge-one} replaced with \cref{small-iid,converge-one-iid}, respectively.
\end{proof}

\begin{proof}[Proof of \cref{empirical-iid}]
Analogous to the proof of \cref{empirical},
with $S^n_{\varepsilon/(n+2)}$ replaced with $\lfloor S\rfloor_{\varepsilon/3}$,
and with \cref{elkind,efficient,uniform} replaced with \cref{elkind-iid,exists-and-efficient-iid,uniform-iid}, respectively.
\end{proof}

\end{document}